%% file: arxiv.tex
\documentclass[11pt]{article} 
\usepackage[margin=1in]{geometry}
\usepackage{helvet}  
\usepackage{courier}  
\usepackage[hyphens]{url}  
\usepackage{graphicx} 
\usepackage{natbib}  
\usepackage{caption} 
\usepackage{amsmath}
\usepackage{amsthm}
\usepackage{amsfonts}
\usepackage{bm}
\usepackage{enumerate}
\usepackage{enumitem}
\usepackage{cleveref}
\usepackage{multirow}
\usepackage{pbox}
\usepackage{xfrac}
\usepackage{thmtools}
\usepackage{thm-restate}
\usepackage{subcaption}
\usepackage{multirow}
\usepackage{algorithm}
\usepackage{algorithmic}

\usepackage{newfloat}
\usepackage{listings}
\DeclareCaptionStyle{ruled}{labelfont=normalfont,labelsep=colon,strut=off} 
\floatstyle{ruled}
\newfloat{listing}{tb}{lst}{}
\floatname{listing}{Listing}
\newtheorem{theorem}{Theorem}[section]
\newtheorem{lemma}{Lemma}[section]

\newtheorem{assumption}{Assumption}

\newtheorem{corollary}[lemma]{Corollary}

\newtheorem{definition}{Definition}
\newtheorem{observation}[lemma]{Observation}
\newtheorem{conjecture}[theorem]{Conjecture}

\title{The Publication Choice Problem}
\author{
    Haichuan Wang\thanks{The work was done when Haichuan Wang was an undergraduate student at University of Chicago.}\\
    Harvard Univrsity\\
\texttt{haichuan\_wang@g.harvard.edu}\\
    \and
    Yifan Wu\thanks{The work was done when Yifan Wu was a PhD student at Northwestern University, supported by NSF ECCS 2216970, under the IDEAL Summer Research Exchange Program. }\\
    Microsoft Research\\
    \texttt{yifan.wu2357@gmail.com}\\
    \and
    Haifeng Xu\thanks{Haifeng Xu acknowledges support from the AI2050 program at Schmidt Sciences (Grant G-24-66104) and the Army Research Office Award
W911NF-23-1-0030. }\\
    University of Chicago\\
     \texttt{haifengxu@uchicago.edu}
}

\input{notation}

\begin{document}

\maketitle

\input{00-Abstract-aaai}

\input{01-Introduction-aaai}
\input{02-prelim-aaai}
\input{03-eq-analysis-aaai}

\input{04-spotlight-aaai}

\input{06-conclusion-aaai}

\bibliographystyle{abbrvnat}
\bibliography{ref}

\appendix
\input{05-empirical-aaai}

\input{appdx-experiment}
\input{appendix_aaai}

\end{document}

%% file: notation.tex
 \newcommand{\model}{Publication Choice Problem}
  
\DeclareMathOperator*{\argmax}{arg\,max}

\newcommand{\venue}{v}

\newcommand{\act}{a}

\newcommand{\cost}{c}

\newcommand{\typespace}{\Theta}

\newcommand{\confavg}{\venue}

\newcommand{\ratiocost}{b}

\newcommand{\util}{u}

\newcommand{\spfrac}{\Omega}

\newcommand{\spcostratio}{r}

\newcommand{\vspcostratio}{\bm{\spcostratio}}

\newcommand{\type}{\theta}

\newcommand{\vtype}{\bm{\type}}
\newcommand{\vcost}{\bm{\cost}}

\newcommand{\vvenue}{\bm{\venue}}
\newcommand{\vact}{\bm{\act}}

\newcommand{\vmu}{\bm{\den}}
\newcommand{\vden}{\bm{\den}}

\newcommand{\uvec}{\bm{1}}

\newcommand{\den}{\mu}

\newcommand{\spadv}{\gamma}

%% file: 00-Abstract-aaai.tex
\begin{abstract}

Researchers strategically choose where to submit their work in order to maximize its impact, and these publication decisions in turn determine venues’ impact factors. To analyze how individual publication choices both respond to and shape venue impact, we introduce a game-theoretic framework—coined the \emph{\model}—that captures this two‐way interplay. We show the existence of a pure-strategy equilibrium in the \model{} and its uniqueness under binary researcher types. Our characterizations of the equilibrium properties offer insights about what publication behaviors better indicate a researcher's impact level.
Through equilibrium analysis,  we further investigate how labeling papers with ``spotlight'' affects the impact factor of venues in the research community. Our analysis shows that competitive venue labeling top papers with ``spotlight'' may decrease overall impact of other venues in the community, while less competitive venues with ``spotlight'' labeling have a opposite impact.


\end{abstract}

%% file: 01-Introduction-aaai.tex
\section{Introduction}
How to choose publication venues is a strategic choice of researchers because they derive rewards from publications—particularly in prestigious venues—which are central to a research career. As a result, researchers are often rational about their publication strategies in response to venue impacts. In turn, the average impacts of publication venues are also subject to the strategic behaviors of researchers and can co-evolve over time with such behaviors. For instance, the rising popularity of machine learning conferences has been associated with a perceived decline in the \emph{average} impact of publications. Consequently, researchers might instead choose smaller conferences that are considered more selective and more impactful.

Inspired by the feedback loop in the job market signaling game \citep{spence1978job}, we propose the \model{} and analyze its dynamics and equilibrium. 
The game consists of a continuum of researchers (agents) and a set of publication venues, defined as follows.  
\begin{itemize}
    \item Researchers (agents). Each researcher in our model is a principal investigator (PI), parameterized by a type representing impact level and endowed with a uniform time budget. The type captures the researcher's productivity, research taste, etc., which we assume is determined by a researcher's past actions and does not vary over the time considered in our model. The uniform time budget assumption reflects that PIs generally have a similar amount of time available for research. Students help convert the PI's time into publications. A large group of students does not increase the PI's time budget but reduces her publication costs.
    \item Publication venues. Each venue has a venue impact and an intrinsic competitiveness level that reflects on the cost of a publication for each researcher type on the venue. The venue impact is defined as the weighted average of researcher types who publish papers on the venue. The publication cost captures the publication venue's acceptance rate, location, preference over topics, etc., and does not vary over the time considered in our model.
\end{itemize} 
Researchers strategize to gain utility from publishing in high-impact venues. The utility of the researcher can result from the recognition of her paper by a venue where high-impact type researchers publish. 
Each researcher solves a utility maximization problem, with the actions being the number of publications in each venue, and subject to the constraint of total time budget. For example, consider an \model{} with two venues, venue $1$ and $2$. The researcher observes  impacts of the two venues as $\confavg_{1} = 0.2$ and $\confavg_{2} = 0.7$. Assuming the researcher uses up her time budget, she can   choose to publish either (A)  $3$ papers on venue $1$ and $1$ paper on venue $2$, or (B) $1$ paper on venue $1$ and $2$ papers on venue $2$. The utility gained from strategy (A) is lower than that of strategy (B), leading the researcher to prefer the latter. 
\begin{table}[htbp]
\centering
\label{tab:two-act}
\begin{tabular}{|c|c|c|c|}
\hline
 &\multicolumn{2}{c|}{\# publications}  & \multirow{3}{*}{Utility} \\
\cline{2-3}
& venue $1$ & venue $2$ & 
\\
\cline{1-3}
Impact & 0.2 & 0.7 & \\
\hline
Action 1 & 3 & 1 & $3 \cdot 0.2 + 0.7 = 1.3$ \\
\hline
Action 2 & 1 &  2 & $2 \cdot 0.7 + 0.2 = 1.6$ \\
\hline
\end{tabular}
\caption{The utility of two actions. Each action is a vector of the number of publications (\# publications) on each venue.}
\vspace{-3mm}
\end{table}

Upon observing venues' impacts from the last round, the researchers modify their publication strategies. In each round in our model, researchers simultaneously choose the number of papers to publish in each venue upon observing the impact factor of all the venues. After all researchers publish their papers, the impact of each venue is updated to the average type of researchers who publish in that venue. These updates are then observed by the researchers in the next round.
An equilibrium is reached if venue impacts do not change over rounds, i.e., when venue impacts align with researcher actions. 

By modeling and analyzing the \model{}, we are able to reveal the effect of researchers' strategic publication behaviors on the publication venues' impact. 
We summarize our main results as below.
\begin{description}
    \item[Equilibrium Existence] A pure-strategy equilibrium always exists (\Cref{prop:eq-existence}). Moreover, under the binary-type setting (researchers are either high-impact or low-impact), there exists a unique pure-strategy equilibrium (\Cref{thm:2p-uniqueness}).
    \item[Indicators of researcher's impact level] The total number of publications is not monotone in researcher's impact (\Cref{observation:low-more-total-pub}). However, the number of publications on the best venue is monotone in researcher's impact (\Cref{thm:top-venue-monotone}).
    \item[Spotlight Labeling] When a venue labels some papers as ``spotlight'' papers, there exists a threshold effect on the venue impact (\Cref{thm:sp-impact-compare}): a less competitive venue improves the impact of all venues by setting up spotlight labeling, while a more competitive venue decreases the impact of all venues. In fact, the spotlight papers divert the impact in a research community from regular venues to the special spotlight papers, thus reducing overall impacts. 
\end{description}

The rest of the paper proceeds as follows. In \Cref{sec: Model and Notations}, we introduce the model of \model. In  \Cref{sec:analysis}, we characterize equilibrium properties. \Cref{sec:best-response} analyzes the best-response dynamic. \Cref{subsec:cost assumption} introduces our key assumption on publication costs, provides justification for it, and offers sanity checks of the model. \Cref{sec:eq analysis} shows the equilibrium properties, including the existence of an equilibrium and the signaling effect of publication numbers in researcher types. We focus on the binary-type setting in \Cref{sec:binary}. We are able to show the equilibrium uniquely exists in \Cref{subsubsec:binary-type}. With the equilibrium uniqueness, in \Cref{sec: spotlight-labeling}, we focus on the venue organizer's spotlight design problem and examine how switching to spotlight labeling may change the equilibrium outcome. 
In Section \ref{sec:future work}, we summarize our contribution and discuss future work. 

\subsection{Related Work}
\paragraph{The Science of Science. } Our paper contributes to the ``science of science'' literature, which studies the science of research publication such as  quantifying research impacts \cite{wang2013quantifying, frank2019evolution, perry2016count}, designing better review systems \citep{su2021truthful,zhang2024eliciting,stelmakh2021peerreview4all,lipton2019troubling, aziz2023group, boehmer2022combating, fromm2021argument, payan2022will} and understanding incentives/competition in academia \citep{zhang2022system,heckman2020publishing,ductor2020influence, manzoor2021uncovering, stelmakh2021catch}. There exists an extensive literature that models strategic behavior in research publication and attempts to understand  and improve peer review processes \citep{zhang2022system, shah2021systemic,wu2023isotonic,zhang2024eliciting,stelmakh2021peerreview4all,lipton2019troubling,wright2015mechanical, jecmen2020mitigating, meir2021market}. Our work also studies strategic behavior, but our focus is mainly on the impact on the research community with multiple publication venues. Along this line, the most relevant to us is  the concurrent work of \citet{ductor2020influence}, which similarly studies the evolution of publication venues in a research community. However, the goal of \citet{ductor2020influence} is completely different from ours, hence their detailed modeling. Motivated by publication conventions in economics, \citet{ductor2020influence} study the evolution of a research community with a single general-purpose publication venue and several field venues. 
However, motivated by the publication patterns in machine learning (or AI, generally), our work does not focus specifically  on the impact of a general-purpose venue\footnote{ML conferences include diverse topics, making research fitness  less important for AI researchers.}. 
Instead, our model considers the impact tiers of different venues, which applies to publication patterns in a general research field, such as sub-fields in computer science and economics. Despite similarities in some model components (e.g., the two-side model structure and   venue quality modeling), on modeling choice, \citet{ductor2020influence} models discrete submission strategy of general-purpose or field venues, with multiple equilibria arising from this discrete strategy space. Our work models the submission numbers to different venues as a vector in a continuous submission space. We are thus able to characterize the closed-form best response.  
 Besides, the data sets and technical results are   different.

\paragraph{Large Population Games. } Our \model{} is a game with a large population, which has been studied in multiple research fields including mean field games \cite{lasry2007mean,aumann1975values, lauriere2022scalable, perrin2020fictitious} and (non-atomic) congestion games \cite{milchtaich1996congestion, friedman1996dynamics, blonski1999anonymous, roughgarden2004bounding}.  Conceptually, of particular relevance is the congestion games where too many players picking a certain action will render that action  bad due to congestion. Our \model \ bears some similarity but differs fundamentally in at least two key aspects, which renders techniques in congestion games (e.g.\ potential function) inapplicable to our problem. First, the utility from each venue does not depend on the total number of players but rather on their average impact.
Second and more importantly, \model{} allows heterogeneous player types while congestion game does not. 

%% file: 02-prelim-aaai.tex
\section{A Model of Publication Choice }
\label{sec: Model and Notations}

\textbf{Conventional Notations.}  Throughout the paper, we use the following conventional mathematical notations. For any matrix $\vact \in \mathbb{R}^{m \times n}$, we use $\vact_{i} \in \mathbb{R}^n$ to denote  $i$-th row vector and $\vact_{:, j} \in \mathbb{R}^m$ to denote the $j$-th column vector. The ``$\cdot$'' is used for inner product of vectors, whereas ``$\odot$'' is for the Hadamard (entry-wise) product, i.e., $\bm{a}\odot \bm{b} = (a_1b_1, a_2b_2,\cdots)$. For notational convenience, we often write inner product as $\bm{a}\cdot \bm{b} $ without explicitly using the transpose notation.

\noindent We consider a game with a unit of continuum researchers and finitely many publication venues, denoted by set $V = \{1, \cdots, k\}$. 
Any researcher is characterized by a \textit{type} $\theta \in \mathbb{R}$, the researcher's impact factor, determined by her publication history, research taste, productivity, etc. 
Let $\Theta \subset \mathbb{R}$ denote the set of all possible researcher types, which is assumed to be discrete. A generic researcher type is denoted as $\theta_i \in \Theta$ and has mass $\mu(\theta_i)$. We sort the parameters $\theta_1 < \theta_2 < \cdots < \theta_n$ increasingly.   We assume a researcher's type $\type_i$ does not change over the short term considered in our \model{}, an assumption justified by a simulation in which we initialize the venues' impact at 50 random starting points. In all cases, the simple best-response dynamics converge rapidly—within only a few iterations—to the equilibrium of the game (\Cref{fig:fast_convergence}). The full experimental setup is described in Appendix~\ref{appdx:fast-convergence}.

\begin{table}[htbp]
\vspace{-2mm}
    \centering
    \begin{tabular}{c|c|c|c|c}
     \# Rounds to Convergence    & 4 & 5 & 6 & 7 \\
     \hline
      Proportion   & 2\% & 8\% & 86\% & 4\%\\  
      \hline
    \end{tabular}
    \vspace{-1mm}
    \caption{Rounds before converging to equilibrium. }
    \label{fig:fast_convergence}
    \vspace{-2mm}
\end{table}


We consider symmetric strategies in the sense that any two researchers of the same type select the same publication strategy. Let  vector $\bm{a}_i \in \mathbb{R}^k$ denote type $\theta_i$'s strategy profile, where each entry 
$\act_{i, j}$ is the number of publications that a type $\type_i$ researcher publishes on venue $j$.  We write the strategy of all researcher types as matrix $\vact$. 

In our model, a venue impact is tied to the types of researchers who publish on the venue. Before reading each paper in detail, the community tends to pay attention to venues which high-impact researchers publish on. Thus, under strategy profile $\vact$, the venue impact $\confavg_j$ is modeled as the weighted average type $\confavg_{j}= \frac{(\vact_{:, j}\odot\vden) \cdot \vtype}{(\vact_{:,j}\odot\vden)\cdot \uvec}$ of researchers who publish on venue $j$, where researcher types are weighted by their number of publications on venue $j$.

We model the cost for researcher $i$ to publish one paper on venue $j$ as $\cost_{i, j}$. 
Venues are sorted by their competitiveness in publication. 
Naturally, a more competitive venue is assumed to have a higher publication cost for every type of researchers.\footnote{We do not model paper acceptance or rejection in our model.  When the acceptance/rejection decision is less random, the cost is naturally deterministic. When the acceptance/rejection decision shows a higher degree of inconsistency and arbitrariness, such as in NeurIPS \citep{cortes2021inconsistency}, the cost captures the expected cost for one publication. As long as the researcher pays enough cost to surpass the venue's basic quality requirement, she can exploit the randomness in acceptance by submitting the same paper to similar venues until the paper gets accepted.} 
Thus, $\cost_{i, j}$ increases in venue index $j$ for each researcher of type $i$. 

A \model{} hence is specified by a tuple $(V, \typespace, \den, \vcost)$, with meanings summarized below.  
\begin{table}[htbp]
    \centering
    \begin{tabular}{c|c}
    \hline
$\vtype=(\type_i)_i  $  & researcher types (impact factors)\\
\hline
$\vden=(\den_i)_i$ &  density on each type $i$ in the community\\
    \hline
$\vact=(\act_{i, j})_{i, j}$         & type $i$'s  num of publications on venue $j$ \\
         \hline
  $\vcost = (\cost_{i, j})_{i, j}$       & type $i$'s cost of publication   on venue $j$\\
         \hline
$\vvenue=(\venue_j)_j$       &  average impact of each venue $j$ \\
         \hline
    \end{tabular}
    \vspace{-1mm}
    \caption{Notations for an \model{}.}
    \label{tab:notations}
    \vspace{-2mm}
\end{table}

 The researcher derives utility from publishing on venues with higher impacts.  
Formally, fixing venue impacts $\vvenue$, the researcher of type $\type_i$ publishes $\vact_i=(\act_{i, j})_j$ on each venue $j$ and  gains utility 
$
    \util_i(\vact_i, \vvenue)= ( 
\vact^{\alpha}_i \cdot \vvenue^{\beta})^{\frac{1}{\beta}}.$
Following an axiomatic characterization about   how to count research impact due to \citet{perry2016count},  we assume the utility function takes the form of a $\beta$-norm. Specifically, \citet{perry2016count} proves that  if a researcher's impact function over publications  satisfies monotonicity, independence, depth  relevance, and scale invariance, then the   only function format is the $\beta$-norm on the vector of citation numbers.  Parameters $\alpha, \beta$ decide the marginal utility gained from a publication and are homogenous across researchers. In the utility function, $\vact_i^\alpha$ captures how a researcher normalizes and counts publications, while $\vvenue^{\beta}$ calculates the utility gained from each publication count. 
We make the following natural assumptions on parameters $\alpha, \beta$: 
\begin{itemize}
    \item $\alpha \in (0,1)$, meaning that the researcher normalizes publication counts on the same venue in a marginally decreasing way.
    \item $\beta>1$, meaning that the utility $\vvenue^{\beta}$ from one publication count is marginally increasing in the venue impact. The reason for this is that a researcher is known for and typically cares more about her most impactful works. 
\end{itemize}

\paragraph{Best Responses.} Upon observing venue impacts $\vvenue$, researchers strategize to maximize the utility from publication. The publication game in real life evolves dynamically: the researchers observe the venue impacts from the previous round and decide their publication strategies for the current round; the venue impacts are updated from researchers' strategies and are observed again; a new loop of researcher best response starts. We define the best response for each researcher: fixing the venue impacts, each researcher  $i$  optimizes her publication utilities, subject to a total cost budget constraint, in Program \eqref{prog:best-response}. We normalize the time budget of all researchers to $1$. Note that this normalization also works for PIs with a large research lab, since a larger lab does not increase the PI's budget but lowers publication cost.\footnote{Program \eqref{prog:best-response} admits the same optimal solution as when the objective is the utility $(\vact^{\alpha}_i \cdot \vvenue^{\beta})^{\frac{1}{\beta}}$.}
\begin{align}
\max_{\vact_i}  \qquad
&\vact^{\alpha}_i \cdot \vvenue^{\beta}
\qquad \text{s.t. }\, \vact_i \cdot \vcost_i \leq 1
\label{prog:best-response}
\end{align}

\vspace{-2mm}
\paragraph{The Equilibrium. } We study the following natural   equilibrium concept with a continuum of researchers, adapted from \citep{MV-93}. With a continuum of researchers, the venue impacts are unaffected by any single researcher’s strategy. The equilibrium condition requires that when all researchers are best responding, their strategies are consistent with the venue impacts. An equilibrium can be viewed as a fixed point in the dynamic loop consisting of best responses. 

\begin{definition}[Equilibrium with a Continuum of Researchers] 
For any \model{} $(V, \typespace, \den, \vcost)$ 
, an action profile $\vact = \{ \act_{i, j} \}_{i,j}$ and venue impacts $\vvenue$ are a pure-strategy equilibrium if they satisfy the following two conditions:
\begin{itemize}
    \item \textbf{[Best Response]} The strategy $\vact^*_{i}$ of  any type-$\theta_i$ researcher  is a best response, i.e.\ solves Program \eqref{prog:best-response}. 
    \item \textbf{[Consistency of Venue Impact]} The venue impacts are consistent with researchers' strategies:  
    \begin{align}
         \venue_j = \frac{(\vact^*_{:, j} \odot \vmu) \cdot  \vtype}{(\vact^*_{:,j}\odot \vmu) \cdot \uvec}, \, \, \,  \text{ for any venue } j 
     \label{eq:equilibrium-venue-consistency}
    \end{align}
\end{itemize}
\end{definition}

%% file: 03-eq-analysis-aaai.tex
\section{Properties of  \model s}
\label{sec:analysis}

We analyze the equilibrium properties of 
the    \model{}. In \Cref{sec:best-response}, we derive the closed form of the best response problem for each researcher. 
In \Cref{sec:eq analysis}, we derive properties of \model{} at equilibrium. 

\subsection{ Characterizing Researchers' Best Responses}
\label{sec:best-response}

We start by characterizing a researcher's best response to the venue impacts as a solution to Program \ref{prog:best-response}. As shown in \Cref{prop:best-response}, 
 in the best response strategy of each researcher, the published amount of  papers is proportional to 1) the utility 
$\venue_{j}^{\frac{\beta}{1-\alpha}}$ of publication count on each venue, and 2) the marginal cost $(\cost_{i,j})^{\frac{1}{\alpha-1}}$ of one normalized publication count $\act_{i, j}^{\alpha}$ in utility.

\begin{restatable}{lemma}{propbestresponse}\label{prop:best-response}
    Let  $\vvenue$ be the venue impact vector. Then the best response of any  researcher of type $\type_i$ can be characterized in closed-form as follows:  
\begin{align*}
    \act_{i,j} &= \frac{(\cost_{i,j})^{\frac{1}{\alpha-1}} \cdot \venue_{j}^{\frac{\beta}{1-\alpha}}}{\vcost_{i}^{\frac{\alpha}{\alpha-1}} \cdot \vvenue^{\frac{\beta}{1-\alpha}}}
\end{align*}
\end{restatable}
The proof of \Cref{prop:best-response} derives the closed form solution for Program \ref{prog:best-response} and is deferred to \Cref{appdx:best-response}.

\subsection{Natural Properties and Model Sanity Check}
\label{subsec:cost assumption}
As a warm-up, this subsection exhibits a few natural properties of the equilibrium hence also serves as a sanity check for the validity of our model above before we dive into more evolved analysis afterwards.

Before analyzing the equilibrium of the \model{}, we introduce the following key assumption, the \textit{Monotone Cost Ratio} (MCR). Specifically,  we assume the relative cost between a low type and a high type increases at higher/better venues. 
 
\begin{assumption}[Monotone Cost Ratio (MCR)]\label{assumption:increasing_cost_ratio} 
   The cost ratio between a low type and a high type increases with the venue index $j$, i.e.\ for all types $\type_i<\type_{i'}$ of researcher, and all venues $j<j'$, $  \frac{\cost_{i, j}}{\cost_{i', j}} < \frac{\cost_{i, j'}}{\cost_{i', j'}}$.
\end{assumption} 

Intuitively, the MCR assumption means that high-type researchers gain more \emph{relative} advantage on more selective publication venues. This is a widely adopted assumption  in principal-agent problem \citep{advance_mirco_theory} and mechanism design \citep{bergemann2002information} to describe the advantage of a more skilled agent. In these economic applications, MCR is assumed to capture the intuition that a higher or more qualified agent type will have increasing advantage in generating higher-quality outcomes/products (sometimes also known as \emph{Monotone Likelihood Ratio}, or MLR,  when exerted cost leads to ordered probabilistic outcomes). Notably, MLR is a widely adopted assumption in economic applications (see, e.g., the textbook \citep{advance_mirco_theory}) and even motivated much statistical study on testing MLR properties \citep{karlin1956theory}.

The following \Cref{prop:separable-cost} shows that the MCR property may be an intrinsic reason that  different publication venues often end up having different average impact in reality --- if the ratio $\frac{\cost_{i, j}}{\cost_{i', j}}$ was the same  on different venue $j$, then  all venues will have the same average impact at \emph{any} equilibrium of the game.

\begin{restatable}{proposition}{propseparablecost}
    \label{prop:separable-cost}
    Suppose the cost ratio $\frac{\cost_{i, j}}{\cost_{i', j}} = c(i; i')$ is a constant that is independent of the venue   $v_j$. Then the \model{} admits   a unique equilibrium in which all venues have the same average impact.  
\end{restatable}

The proof of \Cref{prop:separable-cost} is deferred to \Cref{appdx:separable-cost}.

Before proceeding to the main equilibrium analysis, we highlight several natural properties that a reasonable model of publication choice should satisfy. These results serve as a sanity check for our modeling assumptions; formal statements and proofs are deferred to Appendix B:

\begin{itemize}
    \item \textbf{Monotone venue impacts} Under \Cref{assumption:increasing_cost_ratio}, more competitive venues receive higher impact when all researchers best respond to the observed venue impacts (See \Cref{prop:Quality is Monotone in difficulty} in \Cref{appdx:Quality is Monotone in difficulty}).

     \item \textbf{Scale invariance of costs} Scaling the entire cost matrix $\vcost$ by the same positive factor leaves the equilibrium impact unchanged (See \Cref{observation:importance-of-relative-cost} in \Cref{appdx:uniform-scale}). 
     

    \item \textbf{Asymmetric growth of types} In the binary-type setting, increasing the density of high-type researchers raises the equilibrium impact of all venues; increasing the density of low-type researchers lowers it (See \Cref{appdx:scale-eq}).
    
\end{itemize}



\subsection{Equilibrium Existence and Their   Properties}

\label{sec:eq analysis}

We now turn to prove the general existence  of an equilibrium and its various   properties. Moreover, we show that the number of publications on the top venue is monotone in researcher impact type. However, the total number of publications across all venues is not necessarily monotone --- that is,  there exist  simple instances where a researcher with lower impact has a larger number of publications in total.  

We start our study by showing the existence of equilibria.  

\begin{restatable}{proposition}{eqexistence}\label{prop:eq-existence}
   Every \model{} admits a pure-strategy equilibrium. 
\end{restatable}
The proof of \Cref{prop:eq-existence} is deferred to \Cref{appdx:eq-existence}.


Our main finding of this subsection is the following property which shows that the number of publications in the most competitive venue is indeed monotonically increasing in a researcher's type and, moreover,  this property holds in a much stronger sense than merely at equilibrium.  That is, it holds for any $\vact^*_{i} = \argmax_{\vact_i \geq 0:\,  \vact_i \cdot \vcost_i \leq 1} ( \vact_{i})^{\alpha} \cdot \vvenue^{ \beta}   $ that is a best response to some venue impact vector $\vvenue$, regardless $\vvenue$ is an equilibrium vector or not. This in some sense justifies why some of the popular ranking systems (e.g., the \textit{csrankings.org}) choose to rank institutes' impact based a selected set of top venues.   Recall that \Cref{prop:Quality is Monotone in difficulty} says the most competitive venue coincides with the venue with the highest impact. The proof of \Cref{thm:top-venue-monotone} is deferred to \Cref{appdx:top-venue-monotone}.

\begin{restatable}{theorem}{topvenuemonotone}\label{thm:top-venue-monotone}
 Under Assumption \ref{assumption:increasing_cost_ratio}, consider any venue impacts $\vvenue$ and let $\vact^*_{i}$ 
 be researcher type $i$' best responses to $\vvenue$. Then $\act^*_{i,k} > \act^*_{i',k}$ for any two researcher types $i$ and $i'$ with $\type_i > \type_{i'}$ ($k$ is the most competitive venue).  
\end{restatable}





\paragraph{Is the number of publications monotone in impact?}
In the academic realm, publication count is commonly used as one of the proxies for researchers' impact.\footnote{Whether this is a right metric is out of the scope of this paper's research, though we do observe the increasing use of the \textit{csrankings.org} website, which  counts the number of publications at (only) a selected set of top venues,  as a proxy for different institutes' impact in different fields.   } 
This prompts an examination of the extent to which the quantity of publications is related to a researcher's type. The following example show that the total number of publications is generally not monotone in a researcher's impact. 

\begin{restatable}{proposition}{lowmoretotalpub}\label{observation:low-more-total-pub}
    There exists an \model{} such that the total number of publications is not weakly increasing in one's type under the equilibrium. 
\end{restatable}
\begin{proof}[Proof Sketch]
    We construct an example with two researcher types (high and low) and two venues. In this instance, the high type allocates more effort to the more costly competitive venue, resulting in a lower total number of publications. The detailed parameterization of this example is provided in \Cref{appdx:low-more-total-pub}.
\end{proof}

%% file: 04-spotlight-aaai.tex
\section{The Binary-Type Case: Equilibrium Uniqueness and Spotlight Effects}
\label{sec:binary}
In this section, we turn to a fundamental special case of  the binary researcher types, i.e., a high type and a low type. It turns out that in this case, we are able to further show   the uniqueness of its equilibrium. This uniqueness enables clearer analysis of equilibrium properties and the effects of introducing ``spotlight'' acceptance, which has become increasingly popular in today's AI/ML community. Our theoretical results further show how one venue switching to spotlight labeling may lead to unintended consequences on the impact of other venues. 

\subsection{Equilibrium   Uniqueness}\label{subsubsec:binary-type}


  We show the pure-strategy equilibrium is unique when there is a non-competitive venue that randomly or uniformly accepts all papers. For example, all researchers have the option to publish their paper drafts permanently on Arxiv and not on any other conferences or journals. 

\begin{assumption}[Non-competitive Venue]\label{assumption:non competitive venue}
    The least-competitive venue 1 is non-competitive. That is, the cost for publication is the same for all types: $\cost_{i, 1}=\cost_1, \forall i$. 
\end{assumption}

Before proving the uniqueness of equilibrium, we need to introduce the characteristic function of the \model{}. Let $x=\frac{\act_{H, 1}}{\act_{L, 1}}$ and $\tilde{\den} = \frac{\den_H}{\den_L}$. We normalize the impact for the low type to be $1$, and the impact for the high type $\type$. We define the following characteristic function.
\begin{definition}[Characteristic Function] 
\label{def:stability function}
    Given an \model{}, the characteristic function of the choice problem is defined by
\begin{equation*}\label{eq:stability-function}
    f(x)=\sum_l (\cost_{H, l})^{\frac{\alpha}{\alpha-1}} \cdot \confavg_{l}^{\frac{\beta}{1-\alpha}}\left(x-\frac{\cost_{L, 1}}{\cost_{H, 1}}\cdot \ratiocost_l^{-\alpha}\right),
\end{equation*}
\vspace{-1mm}
where  each $\confavg_j(x) = \frac{1 +\ratiocost_j x\type\tilde{\den}}{1+\ratiocost_j x\tilde{\den}}$ with $\ratiocost_j = (\frac{\cost_{H, 1}\cost_{L, j}}{\cost_{L, 1}\cost_{H, j}})^{\frac{1}{1-\alpha}}$.
\end{definition}

\paragraph{Idea of the Characteristic Function} Our characteristic function describes the dynamics in the game after researchers best respond. The venue impacts can be characterized by the action ratio on any one venue when all types best respond, where we take the ratio on the first venue.  
The input to the characteristic function is venue $1$'s current action ratio, and the output is the  change in the action ratio after all researchers best respond. The sign of the function characterizes the direction of change in all venues' impact. The zero point corresponds to an equilibrium, where the action ratios stop updating after best response.

In the following lemma, we summarize four key properties of the characteristic function, and we defer their formal proofs to \Cref{appdx:4_properties}.

\begin{restatable}{lemma}{fstability}\label{lem:f-stability}[Properties of the Characteristic Function] The following four properties about the characteristic function hold. 
    \begin{enumerate}
    \item A binary-type \model{} is in equilibrium if and only if $f(\frac{\act_{H, 1}}{\act_{L, 1}}) = 0$.
        \item If $f(\frac{\act_{H, 1}}{\act_{L, 1}}) < 0$, when researchers update their actions in response to current venue impacts, the impact of all venues will increase after update;
        \item If $f(\frac{\act_{H, 1}}{\act_{L, 1}}) > 0$, when researchers update their actions in response to current venue impacts, the impact of all venues will decrease after update.
        \item     Under \cref{assumption:increasing_cost_ratio} and \cref{assumption:non competitive venue}, the characteristic function  $f$ is convex in $x$. 
    \end{enumerate}
\end{restatable}
With \Cref{lem:f-stability}, we can prove the uniqueness of the equilibrium by showing $f = 0$ admits a unique solution. The full proof is deferred to \Cref{apdx: proof 2p uniqueness}.

\begin{restatable}{theorem}{twopuniqueness}\label{thm:2p-uniqueness}
    Any binary-type \model{} with a non-competitive venue (\Cref{assumption:non competitive venue}) admits a unique pure-strategy equilibrium. 
\end{restatable}

We conjecture the equilibrium is also unique under many-type settings (see  \Cref{appdx: uniqueness many type} for empirical evidence). We leave its formal proof as future work. 
\begin{conjecture}[Uniqueness of equilibrium under many-type setting]\label{conj:many-type-uniquess} 
    We hypothesize that the pure-strategy equilibrium is unique when there is a non-competitive venue that randomly or uniformly accepts all papers. 
\end{conjecture}

Equilibrium uniqueness benefits comparative statics analysis and enables more interpretable policy insights in \Cref{sec: spotlight-labeling}.

\subsection{The Effect of Discriminative Acceptance via Spotlight Labeling}
\label{sec: spotlight-labeling}

In this section, we first introduce a variant of \model{} for the impact of papers with spotlight labeling. Many publication venues nowadays started selectively labeling publications as ``spotlight'' publications. Selected spotlight papers are often tagged in poster sessions, or presented in an oral session in addition to the poster session. The spotlight labeling attracts more attention from the research community, leading to a higher impact gained on selected papers. We analyze the effect of switching to spotlight labeling on the research community. Our analysis of the equilibrium is restricted to a binary-type setting where we can prove the uniqueness of equilibrium for the same reason as in previous sections. We are able to compare the effect of spotlight labeling in equilibrium only when the equilibrium uniquely exists.

Unlike the establishment of a new venue of higher impact, the impact attributed to a spotlight publication is intrinsically linked to the venue's existing impact. Suppose the program committee of venue $j$ decides to separately label some papers as ``spotlight'', with spotlight papers a $\frac{1}{\spfrac_j}$ of the regular papers ($\spfrac_j>1$). Our model assumes that the impact of spotlight papers is decided by the average impact $\venue_j$ on the regular session and fraction $\frac{1}{\spfrac_j}$ of spotlight papers because the majority of the audience are regular session authors. On average, each spotlight paper gains an impact of $\spadv(\spfrac_j)$ times the impact of a paper in the regular session, where the labeling effect $\spadv(\spfrac_j)>1$ is determined by $\spfrac_j$. We use $\spadv(\bm{\spfrac})$ to denote the vector $(\spadv(\spfrac_1), \cdots, \spadv(\spfrac_k))$. In \Cref{appdx:empirical_justification}, we use empirical citation data on CVPR to justify our assumption that $\spadv(\spfrac_j)>1$. Since the spotlight impact is uniquely pinned down by the impact on regular session $\venue_j$ and spotlight ratio $\spfrac_j$, we can express the overall venue equilibrium impact as $\frac{\spfrac_j + \spadv(\spfrac_j)}{1+\spfrac_j}\venue_j$.

We characterize the properties of the equilibrium when a venue switches to spotlight labeling. The design space for the venue organizer is the cost $\cost_{i, j}^S, \forall i$  and the fraction $1/\spfrac_j$ of spotlight publication (by changing paper selection rules).  Let $\vact_{:, j}^S=(\act_{i, j}^S)_j$ be the vector of spotlight publications by all agents on venue $j$. While the average impact of the spotlight papers is $\spadv(\spfrac_j)\venue_j$, the spotlight papers have an actual average impact without spotlight labeling of $\venue_{j}^S=  \frac{(\vact_{:, j}^S\odot\vden)\cdot \vtype}{(\vact_{:, j}^S\odot\vden)\cdot\uvec}$. When choosing from the design space, the organizer faces constraints, including the following. 

\begin{itemize}
    \item The actual impacts of the spotlight papers are higher than regular papers. The actual impacts are the average impact of a hypothetical venue, assuming spotlight papers are selected for this separate hypothetical venue instead of spotlights of the existing venue. 
\begin{equation}\label{constraint:sp-better}
   \frac{(\vact_{:, j}^S\odot\vden)\cdot \vtype}{(\vact_{:, j}^S\odot\vden)\cdot\uvec}> \frac{(\vact_{:, j}\odot\vden)\cdot \vtype}{(\vact_{:, j}\odot\vden)\cdot\uvec};
\end{equation}
\item It is harder to publish a paper labeled ``spotlight'' for all types, i.e.
$
    \vcost_{:, j}^S \geq \vcost_{:, j}.$
\end{itemize}

\begin{table}[htbp]
\vspace{-2mm}
\centering
\label{tab:notation-spotlight}
\begin{tabular}{c|c|c|c}
\hline
 & Venue $j$ & \multicolumn{2}{|c}{ Spotlight $j$} \\
 \hline
Research \\[-2.5ex] Impact & $\venue_{j}$  & \pbox{20cm}{Spotlight Impact \\(after labeling)\\ $\gamma(\spfrac_j) \cdot \venue_{j}$}  & \pbox{20cm}{Actual Impact\\$ \frac{(\vact_{:, j}^S\odot\vden)\cdot \vtype}{(\vact_{:, j}^S\odot\vden)\cdot\uvec}$} \\
\hline
Action & $\vact_{:, j}$ & \multicolumn{2}{|c}{$\vact_{:, j}^S$}\\
\hline
Cost & $\vcost_{:, j}$ & \multicolumn{2}{|c}{$\vcost_{:, j}^S$}\\
\hline
\end{tabular}
\vspace{-1mm}
\caption{Notations for a venue $j$ with spotlight labeling.}
\vspace{-3mm}
\end{table}

We make the following similar monotone cost ratio assumption on the spotlight session due to a similar reason to 
\Cref{assumption:increasing_cost_ratio}. In
\Cref{appdx:justification-harder-spotlight}, we show if publishing a spotlight paper is relatively the same hard as publishing a regular paper, the actual impact of spotlight papers as a new venue will be the same as regular papers. This violates Constraint \ref{constraint:sp-better} that spotlight papers should gain more actual research impact than regular papers on average. \Cref{assumption:harder-spotlight} states that publishing a spotlight paper should be relatively harder for lower types than a regular paper.

\begin{assumption}[Monotone Spotlight Cost Ratio]\label{assumption:harder-spotlight} 
    The relative cost for any low type to publish a spotlight paper is higher than the relative cost to publish a regular paper. i.e.\ for any two types $\type_i<\type_{i'}$, $\frac{\cost_{i, j}^S}{\cost_{i', j}^S}>\frac{\cost_{i, j}}{\cost_{i', j}}$.
\end{assumption}

Intuitively, the expected cost of publishing a spotlight paper can dynamically change with the number of publications. In \Cref{appdx:one-shot-spotlight-cost}, we solve the researcher's utility maximization problem with spotlight session and show that the cost of publishing a spotlight paper can be fixed once the cost of publishing a regular paper is fixed.

\subsubsection{Characterization of equilibrium with spotlight labeling}\label{sec:eq-spotlight-labeling}
In \Cref{sec:eq-spotlight-labeling}, we provide characterizations of the equilibrium after switching to spotlight labeling. We focus on a binary-type \model{} in this section. \Cref{lem:eq-sp-unique} shows the equilibrium is unique. 

\begin{restatable}
    {corollary}{lemequniquespotlight}\label{lem:eq-sp-unique}
Consider any binary-type \model{} with one venue using the spotlight labeling.   Under \Cref{assumption:non competitive venue} and \Cref{assumption:harder-spotlight}, there exists a unique pure-strategy equilibrium. 
\end{restatable}

The proof of \Cref{lem:eq-sp-unique} is deferred to \Cref{appdx:eq-sp-unique}.



The following theorem shows that, if the organizer cares about absolute research impact but not relative impact in the community, then less competitive venues are better off switching to spotlight labeling. Otherwise, if the venue's regular venue is competitive enough, the spotlight session attracts too much research impact that it hurts the average research impact on every venue in the community.

\begin{restatable}
    {theorem}{thmspimpactcompare}\label{thm:sp-impact-compare}
Consider a binary-type \model{}   under \Cref{assumption:non competitive venue} and \Cref{assumption:harder-spotlight}. Then there exists a threshold venue $j_0$ such that 
\begin{itemize}
    \item if a venue $j\geq j_0$ (more competitive) switches to spotlight labeling, the equilibrium   impact of all venues decrease; 
    \item if a venue $j< j_0$ (less competitive) switches to spotlight labeling,  the equilibrium impact of all venues increase.  
    
\end{itemize} 
\end{restatable}

The proof of \Cref{thm:sp-impact-compare} reduces the problem to the scaling effect of the equilibrium. If we view the spotlight publications as attracted to a separate venue, the spotlight venue changes the ratio of remaining types in the community. When the remaining types in the community scale with different proportion, the equilibrium impact of venues shifts monotonically. We defer the proof of \Cref{thm:sp-impact-compare} to \Cref{appdx:threshold_effect_proof}.

\begin{figure}[thbp]
\vspace{-2mm}
    \centering
\includegraphics[width = \columnwidth]{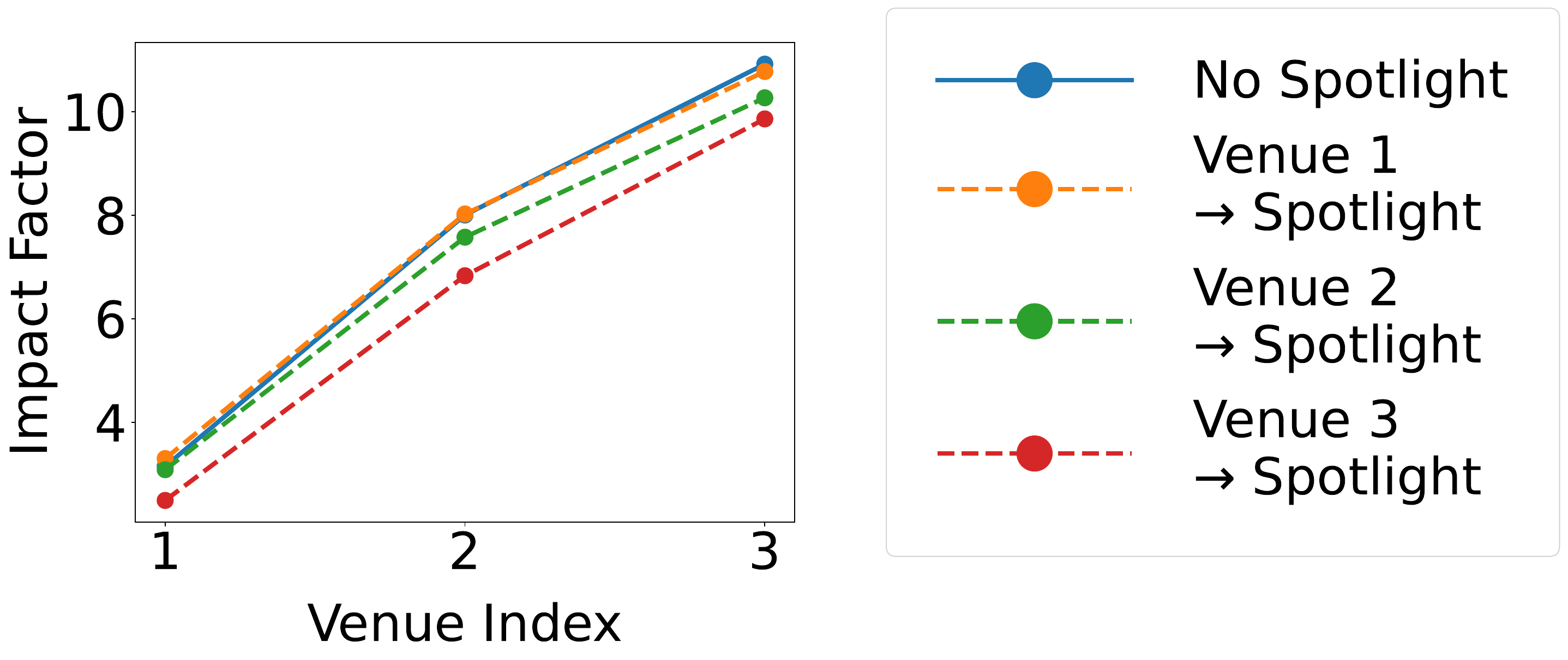}
\vspace{-7mm}
    \caption{Impact factors of regular sessions with and without a venue switching to spotlight labeling. Higher venue indices correspond to more competitive venues.}
    \label{fig:threshold_effect}
    \vspace{-5mm}
\end{figure}

\textbf{Empirical: threshold effect of spotlight labeling holds for many-type settings.}  Under \Cref{conj:many-type-uniquess}, we empirically show the threshold effect in a simulation with five researcher types and three venues; the full setup appears in \Cref{appdx: many-type-threshold}. Figure \ref{fig:threshold_effect} compares the baseline equilibrium impacts (solid blue line) with those obtained when a single venue adopts spotlight labeling (dashed lines). When a more competitive venue introduces spotlight labeling—shown by the red (venue 3) and green (venue 2) dashed lines—the impact factors at all regular sessions fall below the baseline. This occurs because the spotlight session of a competitive venue attracts disproportionately many high-type researchers away from regular sessions. In contrast, when a less competitive venue adopts spotlight labeling—shown by the orange dashed line (venue 1)—the impact at venue 2 rises above the baseline. These results empirically generalize the threshold effect to the many-type setting.

Recall that Constraint \ref{constraint:sp-better} requires that spotlight papers have a higher impact on average. One obvious strategy that organizers may use is to only select papers by researchers with high research impact into the spotlight session, i.e.\ setting costs $\cost_{H, j, S}<\infty$ and $\cost_{L, j, S}=\infty$. We note that this strategy leads to high-type researchers less willing to publish on regular venues, and decreases the impact of all regular venues in the community.

\begin{corollary}\label{observation:only-high-in-spotlight-is-bad}
Under \Cref{assumption:non competitive venue} and \Cref{assumption:harder-spotlight}, for a binary-type \model{} with $k$ venues, where some venue $j$ switches to spotlight labeling, If venue $j$ only labels papers from high-impact type $\type_H$ researchers as ``spotlight'', the equilibrium research impacts of all regular venues will decrease. 
\end{corollary}

The proof follows directly from the proof of \Cref{thm:sp-impact-compare}. \Cref{observation:only-high-in-spotlight-is-bad} implies that the organizer should diversify the set of authors with ``spotlight'' papers.

%% file: 06-conclusion-aaai.tex
\section{Simulations, Conclusions, and Future Work}
\label{sec:future work}

\paragraph{Summary of Simulation Results} We empirically study the \model{} to validate our modeling and results. \Cref{appdx:parameter-setup} describes  simulation setups. \Cref{appdx: uniqueness many type} checks the uniqueness of equilibrium under many-type setting. \Cref{appdx: many-type-threshold} and \Cref{appdx:fast-convergence} describe the experimental setups for the many-type threshold effect and the fast convergence result, respectively. Finally, \Cref{appdx:cost-mat-influence} and \Cref{appdx:different-spotlight-ratio} analyze how varying relative costs and spotlight ratios influence the equilibrium outcomes


\paragraph{Conclusion and Future Work} This paper proposes a game-theoretic model, the \model{}, that explains the interplay between researchers'  publication choice and the evolution of publication venues' impact. We study the properties of the game in equilibrium from an observer's perspective of the research community.

Our results can be divided into two sets: 
 results that an observer of the research community should expect from a game-theoretic model: the equilibrium existence and the scaling effect, which justify our model choice;
 results that shed light on the publication patterns: the monotonicity of publication number in researcher type and the effect of spotlight labeling. 
In future work, we will consider the optimization problem of the venue organizers.


We outline several key observations whose formal results are not included in the current version of the paper, which we identify as future directions.

\begin{description}
    \item[Theoretical Direction] On the theory side, we leave the formal proof for equilibrium uniqueness (\Cref{conj:many-type-uniquess}) and the threshold effect (\Cref{thm:sp-impact-compare}) under many-type setting as future work. 

    \item[Empirical Direction] On the empirical side, our model is not limited to analyzing the AI publication market. \Cref{observation:importance-of-relative-cost} in Appendix suggests that understanding the relative cost of publication in different fields allows our model to predict equilibrium outcomes in other disciplines. One future work is to gather more data and estimate the spotlight advertisement effect $\gamma$. A closed form function $\gamma$ may lead to new theoretical conclusions and practical insights. Given the generality of our model, it could assist a wide range of academic communities in optimizing the impact of their publication venues.
\end{description}

%% file: 05-empirical-aaai.tex
\section{Simulations}
\label{sec:experiments}
In this section, we empirically simulate \model{} to validate our model selection and results. \Cref{appdx:parameter-setup} describes the experimental parameters for all simulations. \Cref{appdx: uniqueness many type} checks the uniqueness of equilibrium and \Cref{appdx: many-type-threshold} validates the threshold effect in a many-type setting, extending our theoretical result under a binary-type setting. \Cref{appdx:fast-convergence} validates that our \model{} converges to an equilibrium very fast, justifying our assumptions that publication costs and researcher types remain constant in the game. 
\Cref{appdx:cost-mat-influence} focuses on assessing the influence of different relative costs on venue impact factor at equilibrium. \Cref{appdx: many-type-threshold} investigates the threshold effect of spotlight labeling in a many-type setting and empirically examines the influence of varying spotlight ratios on venue impact in the equilibrium. \Cref{appdx:different-spotlight-ratio} studies how the equilibrium outcome may change if a venue chooses different spotlight ratios.

\subsection{Parameter set-up for simulation experiments}\label{appdx:parameter-setup}
Parameters listed in \Cref{appdx:parameter-setup} will be used for all experiments in \Cref{sec:experiments}. The time budget for all types of researchers is normalized to $40$, allowing us to interpret each cost matrix cell as the mean weekly hours the PI spends as publication efforts. We assume $\alpha = 0.2$, $\beta = 2$, and the distribution of researcher types follows the pattern outlined in Table \ref{tab:researcher-type-dist}, with more low-type researchers than high-type ones. 

\begin{table}[ht]
\centering
\caption{Researcher Type Distribution}
\label{tab:researcher-type-dist}
\begin{tabular}{|c|c|c|c|c|c|}
\hline
 & Type $1$ & Type $2$ & Type $3$& Type $4$ & Type $5$ \\
 \hline
 $\theta_i = i^2$ & $1$ & $4$ & $9$ & $16$ & $25$\\
\hline
Percentage & $50\%$  & $25\%$ & $15\%$ & $7\%$ & $3\%$ \\
\hline
\end{tabular}
\end{table}

Our simulation stops if the \model{} reaches an $\epsilon$-Nash equilibrium, satisfying Stopping Criterion in \Cref{def:stop-criterion}.
\begin{definition}[Stopping Criterion]\label{def:stop-criterion}
An action profile $\vact = \{ \act_{i, j} \}_{i,j}$ is a pure-strategy equilibrium if it satisfies the following conditions in the simulation
\begin{align*}
    \sqrt{\sum_j (\venue_j - \frac{(\vact_{:, j} \odot \vmu) \cdot  \vtype}{(\vact_{:,j}\odot \vmu) \cdot \uvec})^2} < \epsilon & \qquad\text{($\epsilon$-Nash Criterion)}
   \\ 
    \left((\vact_{i})^{\alpha} \cdot \vvenue^{ \beta}\right)^{\frac{1}{\beta}} \geq \left((\tilde{\vact}_{i})^{\alpha} \cdot \vvenue^{ \beta}\right)^{\frac{1}{\beta}}  \forall i &\qquad \text{(Best Response)}
\end{align*}
\end{definition}

\subsection{Uniqueness of Equilibrium in many-type setting}\label{appdx: uniqueness many type}



We consider \model{} with 5 researchers and 3 venues, but the simulation we develop can be easily extended to more researcher types and venues. 
In the simulation, we take $\epsilon = 10^{-5}$. Here we provide empirical evidence that \model{} has a unique equilibria even in the many-type setting. To verify \Cref{conj:many-type-uniquess} empirically, we fix all parameters of the model and randomly select initial venue impacts from the type space $\typespace \in [1, 25]$. Because of the $\epsilon$-stopping criterion, the simulation doesn't numerically reach the same equilibrium each time. We treat the equilibrium impact of venues as a vector, calculate the Euclidean norm between two simulation instances, and conclude they reach a different equilibrium if the Euclidean norm is larger than $1 \times 10^{-5}$. Otherwise, we say the two instances reach the same equilibrium outcome. Across $50$ simulation instances, the model consistently reaches the \textbf{same} equilibrium outcome. When venues begin spotlight labeling, the model still always converges to the same equilibrium. While formal proof of equilibrium uniqueness is established only for binary-type settings, our empirical findings support the validity of \Cref{conj:many-type-uniquess}. The formal proof for many-type settings is left for future work.

\subsection{Simulation Setup for Spotlight Labeling in Many-Type Settings}
\label{appdx: many-type-threshold}


This section describes the parameter choices used in our many-type simulations of the threshold effect of spotlight labeling.

\paragraph{Spotlight labeling parameter set-up}Venue organizers choose the spotlight ratio $1/\spfrac_j$ and the relative increase in spotlight paper cost (selection rule), i.e., $\spcostratio_{i,j} = \frac{\cost_{i, j}^S}{\cost_{i, j}} \forall i,j$. We assume that 
\begin{equation}\label{eq:spot-selection-rule}
    r_{i,j} = a_j \cdot (N+1-i)^2
\end{equation}
where $a_j > 0$ is a constant. Note, $r_{i,j}$ is decreasing in researcher impact which satisfies Assumption \ref{assumption:harder-spotlight}, and the square term captures the exponential advantage gained by high-impact researchers on spotlight publication. Once the venue organizer fixes the spotlight ratio, one can solve for $a_j$ and thus pin down the paper selection rule by Equation \eqref{eq:sp-frac-cost}. This function form of the selection rule has the advantage of fixing $r_{i,j}/r_{i,j'} \forall i, j, j'$. 

Based on our empirical study of spotlight labeling effect in \Cref{appdx:empirical_justification} using CVPR citation data, we further assume $ \spadv(\spfrac_j) = (\log(\spfrac_j))^p$, where $p > 0$ is a constant. The logarithm reduces the large variance of spotlight ratios we observe in empirical data on CVPR. Note $\spadv(\spfrac_j) > 1$ and is increasing in $\spfrac_j$. Throughout the following analysis, we use $p = 1.77$, the best-fit value for CVPR citation data from $2014$ to $2019$.  We use the cost function specified in Equation \ref{eq-simulation-cost}, and we consider the following cost parameter choice $(z=1, g = 0.6)$, where the relative cost grows fast. We set $1/\spfrac_j = 24\%$, the average spotlight ratio on CVPR from $2014$ to $2019$.

\subsection{Simulation Setup for Convergence of \model{}}\label{appdx:fast-convergence}

This section describes the experimental setup used to assess the convergence speed of \model{}.The simulation parameters are specified in \Cref{appdx:parameter-setup}, and we use the same $\epsilon$-stopping criterion as other experiments (\Cref{def:stop-criterion}). In the simulation, we take $\epsilon = 10^{-5}$. We randomly select initial venue impacts from the type space $\typespace \in [1, 25]$, and we plot the histogram for $50$ instances of \model{} in \Cref{fig:fast_convergence}. The x-axis represents the number of rounds before \model{} converges, and we observe all experiments converge within $7$ rounds for all simulations. This empirically shows the fast convergence speed of \model{}, which justifies our assumption that we only consider \model{} in a short time span. 


%% file: appdx-experiment.tex


\subsection{Influence of Relative Cost on Equilibrium Outcome}\label{appdx:cost-mat-influence}
\Cref{observation:importance-of-relative-cost} reveals that relative cost determines the equilibrium venue impacts. In this section, we investigate the effect of varying relative costs on equilibrium venue impacts under the many-type setting. Observations across disciplines reveal variations in publication strategies and the impact of venues. For instance, researchers in Economics tend to publish fewer papers than their counterparts in Computer Science, and the disparity in publication cost and impact between top-tier and other venues is greater in Economics. This prompts the question: by simulating different relative costs between venues, can we observe varying equilibrium outcomes in \model{} that reflect the diversity of publication patterns across disciplines in reality? To explore this, we assume uniformity in researcher types across disciplines and conduct experiments to evaluate the influence of the relative cost on venue impacts in the equilibrium. 
We present experiments with parameters specified in \Cref{appdx:parameter-setup}. We denote the $(i,j)$-th entry of the cost matrix as follow:
 \begin{equation}\label{eq-simulation-cost}
    \left\{\begin{array}{cc}
  \cost_{i,j} = 1 & \text{if $j = 1$}
     \\
    \cost_{i,j}  =  e^{\frac{g \cdot j^2}{i^2}} & \text{if $j > 1$}
     \end{array}
     \right.
 \end{equation} 
, where $i$, $j$ are the type and venue indexes respectively and $g>0$ is a constant specifying the growth rate of relative cost. We assume the presence of one non-competitive venue across disciplines, exemplified by the Social Science Research Network for social sciences and Arxiv for natural and mathematical sciences. Since \Cref{observation:importance-of-relative-cost} shows that the absolute cost does not impact equilibrium, W.L.O.G we normalize the cost on the non-competitive venue to be $1$. The square term on the exponential power emphasizes the exponential change of relative cost between venues and researcher types. Note that $\cost_{i,j}$ is increasing in $j$, decreasing in $i$, and satisfies Assumption \ref{assumption:increasing_cost_ratio}.
The relative cost ratio is increasing in $g$. We demonstrate the cost parameter choices in Table \ref{tab:parameters-for-cost-matrices}.



\begin{table}[htbp]
\centering
\caption{Parameters for Relative Cost}
\label{tab:parameters-for-cost-matrices}
\begin{tabular}{|c|c|}
\hline
 & Growth rate of relative cost \\
 \hline
 Low & $g = 0.2$\\
 \hline
 Relatively Low & $g = 0.3$\\
 \hline
 Relatively High & $g = 0.4$ \\
 \hline
 High & $g = 0.6$\\
 \hline
\end{tabular}
\end{table}

\textbf{Higher relative cost leads to a larger discrepancy between venue impacts in the equilibrium.} 
Figure \ref{fig:equil-diff-cost} illustrates the equilibrium venue impacts under different relative cost growth rates. As the plot clearly demonstrates, lines with higher relative costs lie above those with lower relative costs on venues 2 and 3, suggesting that disciplines with higher relative costs yield greater impacts at more competitive venues. This is because higher relative costs deter low-type researchers from submitting to competitive venues, enhancing their impact factors.  At the same time, higher relative cost prompts low-type researchers to target less competitive venues, lowering their impact factor as the blue line lies above the other lines on venue 1 in Figure \ref{fig:equil-diff-cost}. Our findings suggest that by calibrating the relative cost parameter $g$ for each discipline, our model can be effectively generalized to study a range of equilibrium outcomes of the \model{} for different disciplines.

\begin{figure}[htbp]
    \centering
\includegraphics[width = 0.45\textwidth]{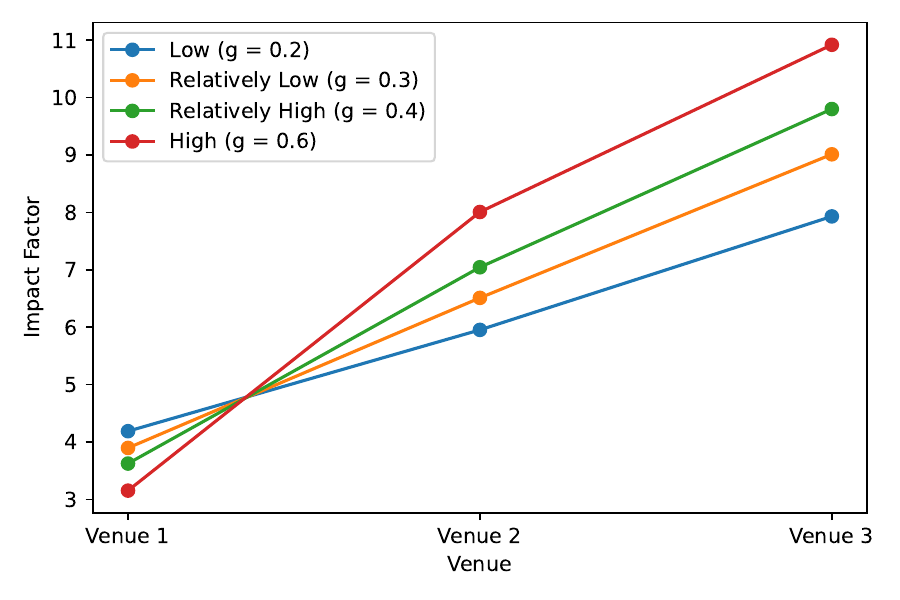}

    \caption{Venue Impacts in the Equilibrium Under Different Relative Cost Growth Rates}
    \label{fig:equil-diff-cost}
\end{figure}

\subsection{Impact of Different Spotlight Ratio} \label{appdx:different-spotlight-ratio}
In this section, we study how the equilibrium venue impacts will change if the venue sets up a different spotlight ratio. We assume the same spotlight selection rule as in \Cref{appdx: many-type-threshold}, and we consider the following cost parameters: ($z = 1,g =0.4$).

\textbf{A lower spotlight ratio will have more influence on the other venues' equilibrium outcome.} We let different venues switch to spotlight labeling, testing spotlight ratio  $5\%$, $25\%$, and $80\%$. In Figure \ref{fig:venue-1-different-sr}, we observe that the blue line  ($5\%$ spotlight ratio) lies below the other two lines at venue 3. This decline demonstrates that a less competitive venue (venue 1 in our example), by setting a very low spotlight ratio, can make its spotlight session comparable to the more competitive venues, thereby attracting high-impact researchers away from the competitive venues. This design choice reduces the impact factor of the regular sessions at more competitive venues. When more competitive venues switch to spotlight labeling, they siphon attention from high-type researchers, diminishing the impact of regular sessions across all venues. In Figures \ref{fig:venue-2-different-sr} and \ref{fig:venue-3-different-sr}, we observe that lines representing lower spotlight ratios are positioned below those with higher spotlight ratios. This decline in impact reveals that the lower the spotlight ratio at competitive venues, the greater the adverse effect it will have on the impact of other venues. Moreover, the gaps between lines are larger in Figure \ref{fig:venue-3-different-sr} compared to those in Figure \ref{fig:venue-2-different-sr}, which suggests that the negative impact of spotlight labeling on other venues intensifies with the venue's competitiveness. This observation makes sense, as the more competitive the original regular session is, the more attention its spotlight session will attract from high-type researchers, diverting their contributions away from other venues.

\begin{figure}[thpb]
\centering
\begin{subfigure}[b]{\columnwidth}
  \includegraphics[width=\linewidth]{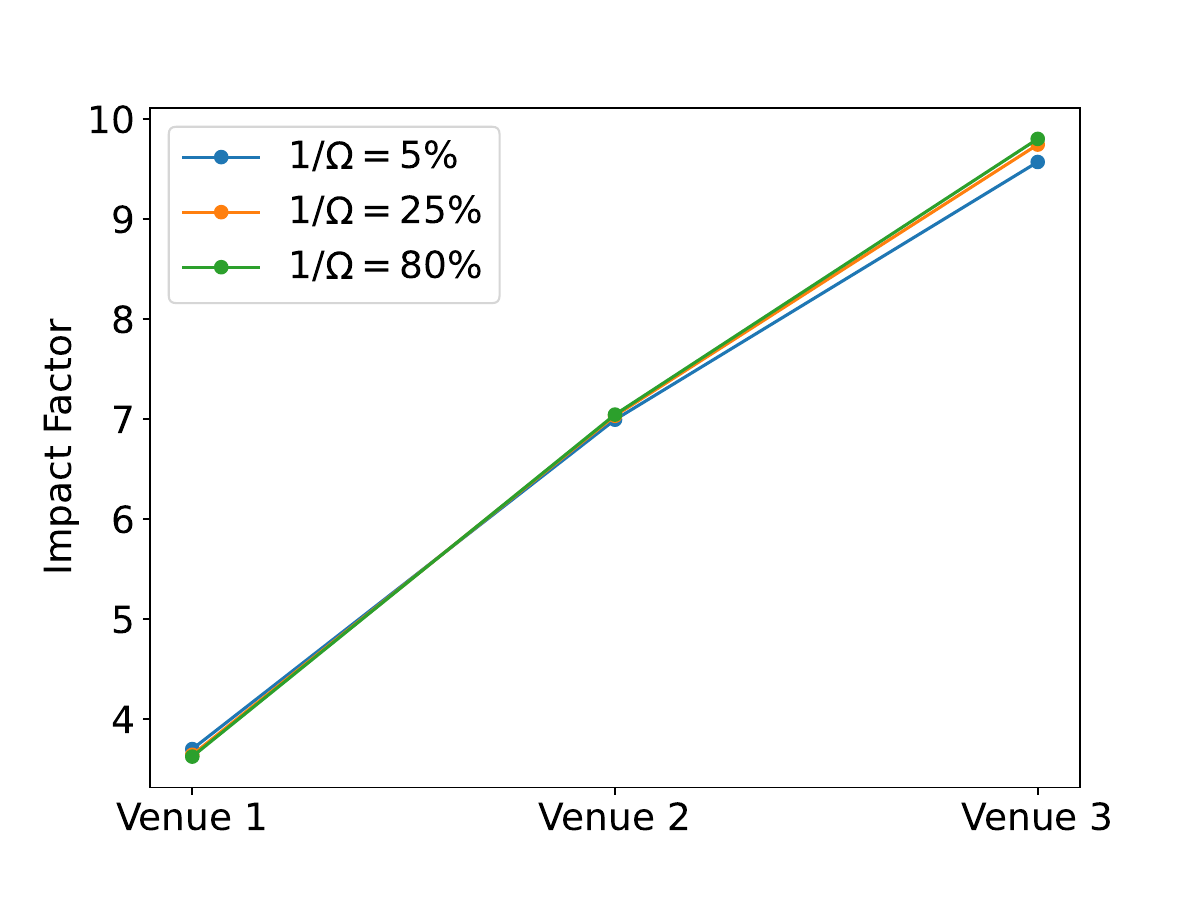}
  \caption{Venue 1 Switches to Spotlight Labeling}
  \label{fig:venue-1-different-sr}
\end{subfigure}
\hfil
\begin{subfigure}[b]{\columnwidth}
  \includegraphics[width=\linewidth]{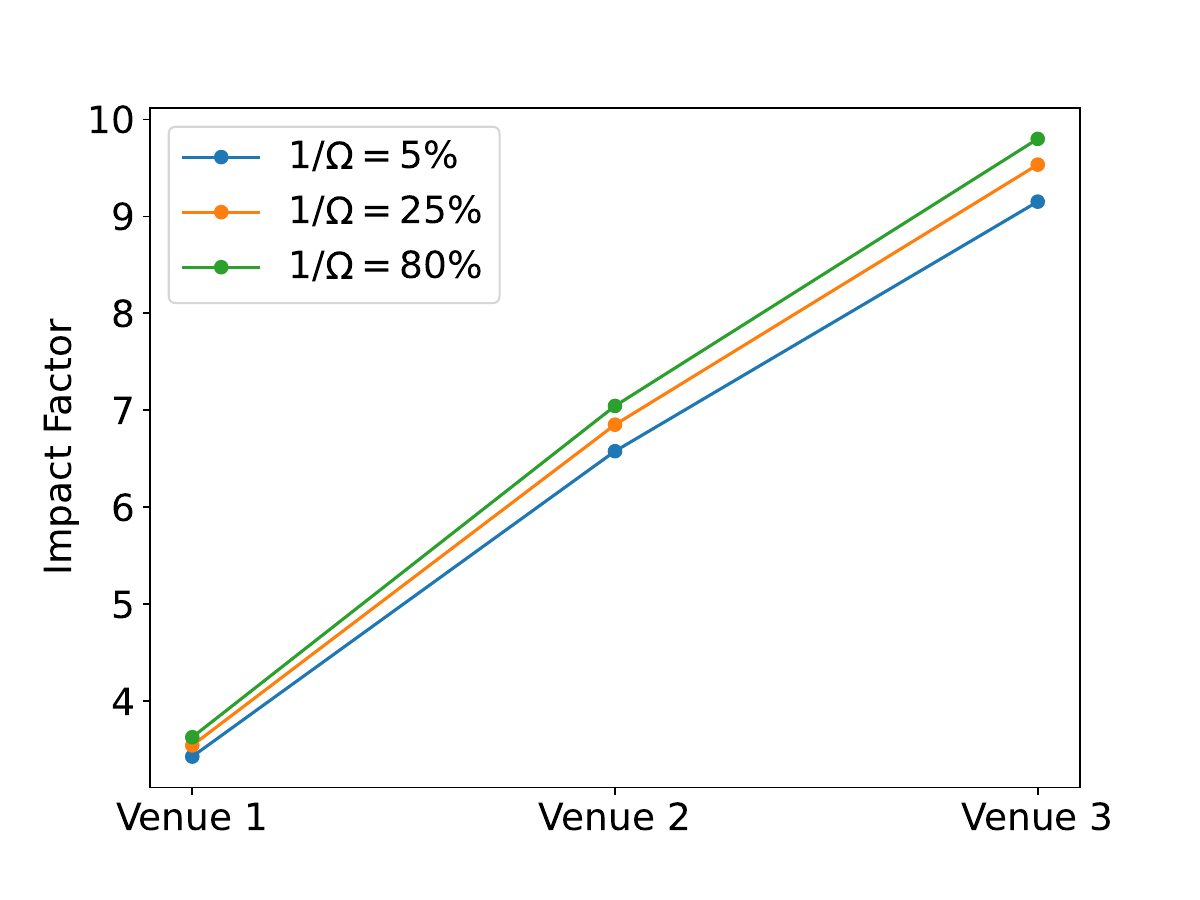}
 \caption{Venue 2 Switches to Spotlight Labeling}
  \label{fig:venue-2-different-sr}
\end{subfigure}%
\hfil
\begin{subfigure}[b]
    {\columnwidth}
  \includegraphics[width=\linewidth]{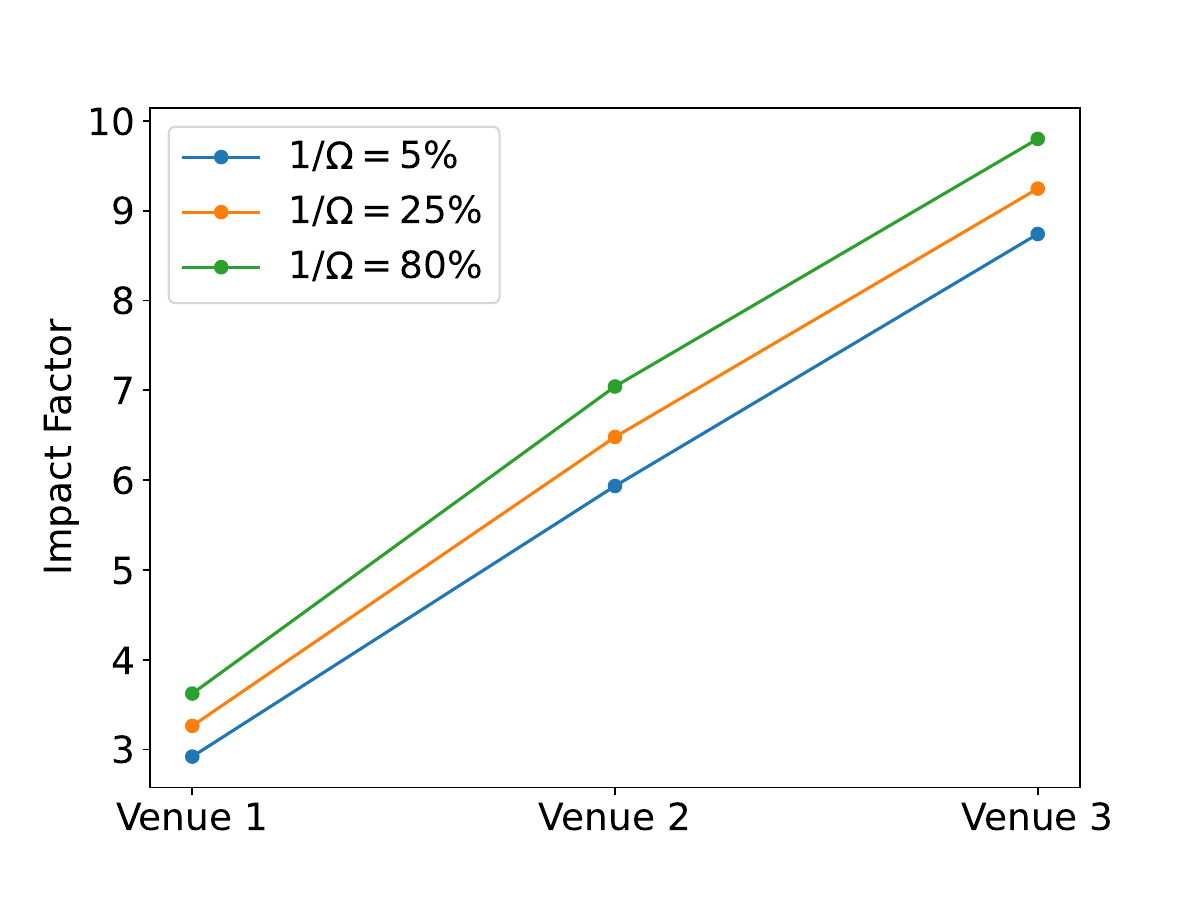}
  \caption{Venue 3 Switches to Spotlight Labeling}
  \label{fig:venue-3-different-sr}
\end{subfigure}
 \caption{Figure \ref{fig:venue-1-different-sr}-\ref{fig:venue-3-different-sr}} show the impact factor of the regular session of each venue in the equilibrium when the venue organizer selects different spotlight ratios
\end{figure}

%% file: appendix_aaai.tex
\section{Missing Proofs}

\subsection{The proof of \Cref{prop:best-response}}\label{appdx:best-response}
\propbestresponse*
\begin{proof}
We prove the proposition by solving each researcher's utility maximization problem using the Lagrangian multiplier method. We write the Lagrangian of the maximization problem in Program \eqref{prog:best-response}: 
\begin{align*}
     L(\vact_{i}, \lambda_i) = (\vact_i)^{\alpha} \cdot \vvenue_j^\beta + \lambda_i \cdot (1-\vact_i \cdot \vcost_i)
\end{align*}  
Because the researcher can always derive marginal utility from publication, they will all choose to exhaust their budget set. Setting the First Order Condition and Budget Constraint we have:  
 \begin{equation}\label{eq-constraints}
    \left\{\begin{array}{cc}
   \alpha \cdot (\act_{i,j})^{\alpha-1}  \cdot \confavg_{j}^\beta - \lambda_i \cdot \cost_{i,j} = 0, \forall j & \text{(First Order Condition)}
     \\
    \vact_{i} \cdot \cost_{i} = 1 & \text{(Budget Constraint)}
     \end{array}
     \right.
 \end{equation}
To solve the maximization problem, we first solve the first-order condition in \eqref{eq-constraints} for any $j$. 
\begin{equation}
    \label{eq:preliminary-action}
    (\act_{i,j})^{\alpha-1} = \frac{\lambda_{i} \cdot \cost_{i,j}}{\alpha \cdot \confavg_{j}^\beta} \notag \Rightarrow
    \act_{i,j} = \bigg(\frac{\lambda_i \cdot \cost_{i,j}}{\alpha \cdot \confavg_{j}^\beta}\bigg)^{\frac{1}{\alpha-1}} 
\end{equation}
Plug the action into the budget condition in \eqref{eq-constraints}:
\begin{align*}
    &\sum_{l=1}^k \bigg((\frac{\lambda_i \cdot \cost_{i,l}}{\alpha \cdot \confavg_{l,t}^\beta})^{\frac{1}{\alpha-1}}) \cdot \cost_l^i\bigg) = 1\\
    &\alpha^{\frac{1}{1-\alpha}}(\lambda_i)^{\frac{1}{\alpha-1}} \cdot (\vcost_{i}^{\frac{\alpha}{\alpha-1}} \cdot \vvenue^{\frac{\beta}{1-\alpha}}) = 1\\
    &\lambda_i = \frac{\alpha}{(\vcost_{i}^{\frac{\alpha}{\alpha-1}} \cdot \vvenue^{\frac{\beta}{1-\alpha}})^{\frac{1}{\alpha-1}}}
\end{align*}
After getting the expression for $\lambda_i$, we plug it back in \Cref{eq:preliminary-action} to get the expression for $\act_{i,j}$. 
\begin{align*}
    \act_{i,j} &= \frac{(\cost_{i,j})^{\frac{1}{\alpha-1}} \cdot \venue_{j}^{\frac{\beta}{1-\alpha}}}{\vcost_{i}^{\frac{\alpha}{\alpha-1}} \cdot \vvenue^{\frac{\beta}{1-\alpha}}}
\end{align*}
\end{proof}

\subsection{The proof of \Cref{prop:separable-cost}}\label{appdx:separable-cost}
\propseparablecost*
\begin{proof}[Proof of \Cref{prop:separable-cost}]
The proof hinges on the following useful \Cref{lem:rank-1-cost}, which characterizes the equivalent conditions of constant relative cost on all venues and cost matrix having rank $1$. 

\begin{lemma}\label{lem:rank-1-cost}
The following three statements are equivalent:
\begin{enumerate}
    \item the cost ratio between a low type and a high type does not change with the venue index $j$, i.e.\ for all types $\type_i<\type_{i'}$ of researcher, and all venues $j<j'$, $\frac{\cost_{i, j}}{\cost_{i', j}}= \frac{\cost_{i, j'}}{\cost_{i', j'}}$;
    \item cost $\cost_j^i$ is multiplicative separable in agent type $i$ and venue index $j$;
    \item The cost matrix of \model{} has rank one. 
\end{enumerate}
\end{lemma}
    
Because the cost $\vcost$ is multiplicative separable in $i$ and $j$, we rewrite the cost matrix as $\vcost = \hat{\vcost}_{\type}\cdot\tilde{\vcost}_{\venue}$, where $\hat{\vcost}_\type=(\hat{\vcost}_{\type_i})_i$, $\tilde{\vcost}_{\venue} = (\tilde{\vcost}_{\venue_j})_j$ are attributes associated with agent types and venues. Consider  equilibrium best response in \Cref{prop:best-response}:
\begin{align}\label{multiplicative-cost-action}
    \act_{i, j} &= \frac{(\cost_{i,j})^{\frac{1}{\alpha-1}} \cdot \venue_{j}^{\frac{\beta}{1-\alpha}}}{\vcost_{i}^{\frac{\alpha}{\alpha-1}} \cdot \vvenue^{\frac{\beta}{1-\alpha}}}\nonumber
    = \frac{(\hat{\cost}_{\type_i}\cdot\tilde{\cost}_{\venue_j})^{\frac{1}{\alpha-1}} \cdot \venue_{j}^{\frac{\beta}{1-\alpha}}}{(\hat{\cost}_{\type_i}\cdot\tilde{\vcost}_{\venue})^{\frac{\alpha}{\alpha-1}} \cdot \vvenue^{\frac{\beta}{1-\alpha}}}\nonumber\\
    &= \frac{1}{\hat{\cost}_{\type_i}} \cdot \frac{\tilde{\cost}_{\venue_j}^{\frac{1}{\alpha-1}} \cdot \venue_j^{\frac{\beta}{1-\alpha}}}{ \tilde{\vcost}_\venue^{\frac{\alpha}{\alpha-1}} \cdot \vvenue^{\frac{\beta}{1-\alpha}}}.
\end{align}

Then, consider the equilibrium venue impacts as in \Cref{eq:equilibrium-venue-consistency}.
\begin{align*}
    \confavg_{j} &= \frac{\sum_{i=1}^n \frac{1}{\hat{\cost}_{\type_i}} \cdot \frac{\tilde{\cost}_{\venue_j}^{\frac{1}{\alpha-1}} \cdot \den_i \cdot \confavg_{j}^{\frac{\beta}{1-\alpha}}}{\tilde{\vcost}_\venue^{\frac{\alpha}{\alpha-1}} \cdot \vvenue^{\frac{\beta}{1-\alpha}}} \cdot \type_i \cdot \den_i}{\sum_{i=1}^n \frac{1}{\hat{\cost}_{\type_i}} \cdot \frac{\tilde{\cost}_{\venue_j}^{\frac{1}{\alpha-1}} \cdot \confavg_{j}^{\frac{\beta}{1-\alpha}}}{\tilde{\vcost}_\venue^{\frac{\alpha}{\alpha-1}} \cdot \vvenue^{\frac{\beta}{1-\alpha}}} \cdot \den_i}\nonumber
    = \frac{(\left(\frac{1}{\hat{\cost}_{\type_i}} \right)_i\odot \vtype) \cdot \vden} {\vtype \cdot \vden}
\end{align*}
and is invariant of venue index $j$.

\end{proof}

\subsection{The proof of \Cref{prop:Quality is Monotone in difficulty}}
\label{appdx:Quality is Monotone in difficulty}

\begin{restatable}{proposition}{monotonevenueimpact}\label{prop:Quality is Monotone in difficulty}
    Under Assumption \ref{assumption:increasing_cost_ratio}, venue impact $\venue_j$ monotonically increases in the venue index $j$ after all researchers best respond to the venue impact in the previous round. 
\end{restatable}

Before we prove \Cref{prop:Quality is Monotone in difficulty}, we need to introduce \Cref{lem:constantactionratio}, which shows the invariance of action ratio across any two venues by two fixed researcher types. 
\begin{lemma}\label{lem:constantactionratio}
    Fix any two researchers with type $\type_i$ and $\type_j$. The action ratio on any two venues is only differed by a constant factor depending on their cost amount on the venues. Formally, 
    ($\frac{\act_{i', j}}{\act_{i, j}}$) / ($\frac{\act_{i',v}}{\act_{i,v}}$) =   
    $\frac{(\cost_{i',j})^{\frac{1}{\alpha-1}}(\cost_{i,v})^{\frac{1}{\alpha-1}}}{(\cost_{i',j})^{\frac{1}{\alpha-1}} \cdot (\cost_{i,v})^{\frac{1}{\alpha-1}}}$
\end{lemma}
\begin{proof}
    By rearranging the best response in \Cref{prop:best-response}, we can write the following two equalities for researcher $i'$'s action on venue $j$ and $v$ respectively,

\begin{align*}
    \vcost_{q}^{\frac{\alpha}{\alpha-1}} \cdot \vvenue^{\frac{\beta}{1-\alpha}} = \frac{c_{q,j}^{\frac{1}{\alpha-1}} \cdot \venue_j^{\frac{\beta}{1-\alpha}}}{a_{q,j}} 
    \text{ \ and \ }
    \vcost_{q}^{\frac{\alpha}{\alpha-1}} \cdot \vvenue^{\frac{\beta}{1-\alpha}} = \frac{c_{q,v}^{\frac{1}{\alpha-1}} \cdot \venue_v^{\frac{\beta}{1-\alpha}}}{a_{q,v}}
\end{align*}

Note the left-hand side of both equalities are the same. Similarly, by rearranging the best response in \Cref{prop:best-response}, we can write the following two equalities for researcher $i$'s action on venue $u$ and $v$ respectively,
\begin{align*}
    \vcost_{i}^{\frac{\alpha}{\alpha-1}} \cdot \vvenue^{\frac{\beta}{1-\alpha}} = \frac{\cost_{i,j}^{\frac{1}{\alpha-1}} \cdot \venue_j^{\frac{\beta}{1-\alpha}}}{\act_{i,j}}
\text{ \ and \ } 
    \vcost_{i}^{\frac{\alpha}{\alpha-1}} \cdot \vvenue^{\frac{\beta}{1-\alpha}} = \frac{\cost_{i,v}^{\frac{1}{\alpha-1}} \cdot \venue_v^{\frac{\beta}{1-\alpha}}}{\act_{i,v}}
\end{align*}
By dividing the action of $i$ and $j$ on the same venue, we obtain the following equality:
\begin{align*}
    \frac{(\cost_{i',j})^{\frac{1}{\alpha-1}}}{\act_{i', j}} / 
    \frac{(\cost_{i,j})^{\frac{1}{\alpha-1}}}{\act_{i, j}}
    =  \frac{(\cost_{i',v})^{\frac{1}{\alpha-1}}}{\act_{i',v}} / 
    \frac{(\cost_{i,v})^{\frac{1}{\alpha-1}}}{\act_{i,v}}
\end{align*}
After rearranging the terms, we get the desired equation:
\begin{align}\label{constant-ratio}
  (\frac{\act_{i', j}}{\act_{i, j}}) / (\frac{\act_{i',v}}{\act_{i,v}}) =   
    \frac{(\cost_{i',j})^{\frac{1}{\alpha-1}}(\cost_{i,v})^{\frac{1}{\alpha-1}}}{(\cost_{i',j})^{\frac{1}{\alpha-1}} \cdot (\cost_{i,v})^{\frac{1}{\alpha-1}}}
\end{align}
This completes the proof of the lemma.
\end{proof}


\begin{proof}[Proof of \Cref{prop:Quality is Monotone in difficulty}]
    The proof is a direct comparison between venue impacts. Consider any two venues with indices $j$ and $v$, we want to show that if venue $j$ is more competitive than venue $v$, then venue $j$ will have a higher impact than venue $v$. Formally, if $j > v$, then $\venue_{j} > \venue_{v}$. Consider any two distinct researchers with type $\type_i < \type_{i'}$. Denote the proportion of their types as $\den_i$ and $\den_{i'}$ respectively. By $\Cref{lem:constantactionratio}$, we have:
    \begin{align}
      (\frac{\act_{i', j}}{\act_{i, j}}) / (\frac{\act_{i',v}}{\act_{i,v}}) =   
        \frac{(\cost_{i',j})^{\frac{1}{\alpha-1}}(\cost_{i,v})^{\frac{1}{\alpha-1}}}{(\cost_{i',j})^{\frac{1}{\alpha-1}} \cdot (\cost_{i,v})^{\frac{1}{\alpha-1}}}
    \end{align}
     By Assumption \ref{assumption:increasing_cost_ratio}, we have $\frac{\cost_{i,j}}{\cost_{i',j}} > \frac{\cost_{i,v}}{\cost_{i',v}}$. Plug this inequality back to \ref{constant-ratio}, because $0 < \alpha < 1$ we have:
\begin{align*}
    \act_{i', j} \cdot \act_{i,v}>
    \act_{i, j} \cdot \act_{i',v}
    \end{align*}
Because $\theta_{i'} > \theta_i$, this further implies that:
\begin{align} \label{term-inequality}
      \act_{i', j} \cdot \act_{i,v} \cdot \type_{i'} + \act_{i, j} \cdot \act_{i',v} \cdot \type_i > \act_{i, j} \cdot \act_{i',v} \cdot \type_{i'} +   \act_{i', j} \cdot \act_{i,v} \cdot \type_i
    \end{align}
Inequality \ref{term-inequality} will be the key for proving the theorem later. The impact of venue $u$ and $v$ are
\begin{align*}
\frac{(\vact_{:, j} \cdot \vmu) \cdot  \vtype}{(\vact_{:,j}\cdot \vmu) \cdot \uvec}
\text{\ and \ }
   \frac{(\vact_{:, v} \cdot \vmu) \cdot  \vtype}{(\vact_{:,v}\cdot \vmu) \cdot \uvec}
\end{align*}
, respectively. We want to show that venue $u$'s impact is larger, which is equivalent to showing that 
\begin{align}\label{wts}
    &((\vact_{:, j} \cdot \vmu) \cdot  \vtype) \cdot ((\vact_{:,v}\cdot \vmu) \cdot \uvec) > \nonumber \\
    &((\vact_{:, v} \cdot \vmu) \cdot  \vtype) \cdot ((\vact_{:,j}\cdot \vmu) \cdot \uvec)
\end{align}
We rewrite the left-hand side of \ref{wts} as:
\begin{align}
    &\sum_{p = q} \act_{i, j} \cdot \act_{i',v} \cdot \type_{i'} \cdot \den_i \den_{i'}\nonumber\\
    &+ \sum_{i,i': \type_i < \type_{i'}} (\act_{i, j} \cdot \act_{i',v} \cdot \type_i + \act_{i', j} \cdot \act_{i,v} \cdot \type_{i'})\cdot \den_i \den_{i'}\label{rewrite-LHS}
\end{align}
We rewrite the right-hand side of \ref{wts} as:
\begin{align}
    &\sum_{i = i'} \act_{i, j} \cdot \act_{i',v} \cdot \type_{i'} \cdot \den_i \den_{i'}\nonumber   \\    
    &+\sum_{i,i': \type_i < \type_{i'}} (\act_{i, j} \cdot \act_{i',v} \cdot \type_{i'} + \act_{i', j} \cdot \act_{i,v} \cdot \type_i) \cdot \den_i \den_{i'}\label{rewrite-RHS}
\end{align}
The first terms in \ref{rewrite-LHS} and  \ref{rewrite-RHS} cancel out with each other. Invoke \ref{term-inequality}, we know that $\forall p, q$ with $\type_i < \type_{i'}$ we have:
$$\act_{i', j} \cdot \act_{i,v} \cdot \type_{i'} + \act_{i, j} \cdot \act_{i',v} \cdot \type_i > \act_{i, j} \cdot \act_{i',v} \cdot \type_{i'} +   \act_{i', j} \cdot \act_{i,v} \cdot \type_i$$
Hence, inequality \ref{wts} is true. Therefore, the impact of venue $j$ is higher than the impact of venue $v$. Since $j$ and $v$ are arbitrarily chosen, it satisfies that the venue impact $\venue_j$ monotonically increases in venue index $j$. 
\end{proof}

\subsection{The effect of uniformly scaling-up the community size}\label{appdx:uniform-scale}
Given the fast-growing AI/ML community today, a natural question one may have is how the venue impact would change as the number of researchers grows in the field.  \Cref{prop:u-scale-eq} shows that if each researcher type scales linearly by the same factor, the venue equilibrium impacts remain the same. We discuss the effect of non-uniform scaling for binary-type setting in \Cref{thm:scale-eq}. 

\begin{observation}\label{prop:u-scale-eq}
    The equilibrium impact of venues remains the same when the number of researchers of different types is simultaneously scaled by the same factor $m > 0$. 
\end{observation}
\begin{proof}
    After scaling all types by the same factor, the density of each type does not change. Hence, the equilibrium impact of venues remains the same.
\end{proof}

Scaling all entries in $\vcost$ uniformly is the same as scaling all researchers' budgets, and is also equivalent to scaling their community size proportionally, which does not affect equilibrium outcome by \Cref{prop:u-scale-eq}. In addition, the following \Cref{observation:importance-of-relative-cost} reveals that the relative cost, rather than the absolute cost, determines the equilibrium outcome of \model{}.

 
\begin{restatable}{observation}{importance-of-relative-cost}\label{observation:importance-of-relative-cost}
The equilibrium impact of venues remains the same when the cost matrix $\vcost$ is scaled by the same factor.
\end{restatable}

\subsection{The proof of \Cref{prop:eq-existence}}
\label{appdx:eq-existence}
\eqexistence*
\begin{proof}
We prove the proposition via the Brouwer fixed-point theorem. Define
\begin{align}
    f(\vvenue) = \bigg( \frac{(\vact_{:, 1}\odot\vden) \cdot \vtype}{(\vact_{:,1}\odot\vden)\cdot \uvec}, \cdots, \frac{(\vact_{:, k}\odot\vden) \cdot \vtype}{(\vact_{:,k}\odot\vden) \cdot \uvec} \bigg)
\end{align}

where $\vact_{:, j}$ is researchers' best-response strategy on venue $j$ after observing $\vvenue$.
Intuitively, the function $f$ updates the venues' impact factors after collecting researchers' publication strategies. Consider the function on the space of possible researcher impact level, $f$ on $\{\typespace|\min\{\type_1, \cdots, \type_n\} \leq \typespace_j \leq \max\{\type_1, \cdots, \type_n\}\}$. The space is closed and bounded in $R^n$, and thus it is compact. Moreover, the space is also convex. $f$ is continuous. Therefore, by Brouwer's fixed-point theorem, there exists a fixed point where $f(\vvenue) = \vvenue$. Because researchers' publication strategy reaffirms the impact level of all venues, the system has reached the equilibrium where all agents are best responding.
\end{proof}

\subsection{The proof of the properties of the characteristic function (\Cref{lem:f-stability})}\label{appdx:4_properties}
\fstability*
\begin{proof}
    We prove each property separately.
    \paragraph{Proof for property 1} Proving the first property is equivalent to showing the following two key points:
    \begin{enumerate}
        \item the zeros of $f(x)$ satisfies the equilibrium condition of the ratio of publications on the first venue
        \item the equilibrium condition of the ratio of publications on the first venue implies the equilibrium condition of \model{}
    \end{enumerate}

    We prove by looking at the equilibrium condition of the ratio of publications on the first venue, where the LHS and the RHS of the equation correspond to the ratio of publications before and after the impact update,respectively:
    \begin{equation}\label{eq:eqm-action-conf0}
        \frac{\act_{H, 1}}{\act_{L, 1}}=\left(\frac{\cost_{H, 1}}{\cost_{L, 1}}\right)^{\frac{1}{\alpha-1}}\left(\frac{\sum_l (\cost_{L, l})^{\frac{\alpha}{\alpha-1}} \cdot \confavg_{l}^{\frac{\beta}{1-\alpha}}}{\sum_l (\cost_{H, l})^{\frac{\alpha}{\alpha-1}} \cdot \confavg_{l}^{\frac{\beta}{1-\alpha}}}\right)
    \end{equation}
    Multiply both sides of \Cref{eq:eqm-action-conf0} by $\sum_l (\cost_{H, l})^{\frac{\alpha}{\alpha-1}} \cdot \confavg_{l}^{\frac{\beta}{1-\alpha}}$ we obtain:
    \begin{equation}\label{eq:eqm-action-conf0-next}
        \frac{\act_{H, 1}}{\act_{L, 1}} \cdot \sum_l (\cost_{H, l})^{\frac{\alpha}{\alpha-1}} \cdot \confavg_{l}^{\frac{\beta}{1-\alpha}} = \left(\frac{\cost_{H, 1}}{\cost_{L, 1}}\right)^{\frac{1}{\alpha-1}} \cdot \sum_l (\cost_{L, l})^{\frac{\alpha}{\alpha-1}} \cdot \confavg_{l}^{\frac{\beta}{1-\alpha}}
    \end{equation}
     Let $x=\frac{\act_{H, 1}}{\act_{L, 1}}$. Move the RHS of \Cref{eq:eqm-action-conf0-next} to the LHS and define 
    \begin{equation}\label{eq:eqm-action-conf0-final}
        f(x) =  x \cdot \sum_l (\cost_{H, l})^{\frac{\alpha}{\alpha-1}} \cdot \confavg_{l}^{\frac{\beta}{1-\alpha}} - \left(\frac{\cost_{H, 1}}{\cost_{L, 1}}\right)^{\frac{1}{\alpha-1}} \cdot \sum_l (\cost_{L, l})^{\frac{\alpha}{\alpha-1}} \cdot \confavg_{l}^{\frac{\beta}{1-\alpha}}
    \end{equation}
    
    With some algebra, we can show \Cref{eq:eqm-action-conf0-final} is exactly the characteristic function defined in \Cref{def:stability function}. Therefore, we have shown that zeros of the characteristic function satisfy the equilibrium condition of the ratio of publications on the first venue. 

    By equilibrium condition, we know that for any venue $j$, the following equation on action ratio is true:
    
    \begin{equation*}
        \frac{\act_{H, j}}{\act_{L, j}}=x\cdot(\frac{\cost_{H, 1}\cost_{L, j}}{\cost_{L, 1}\cost_{H, j}})^{\frac{1}{1-\alpha}}
    \end{equation*}

    It implies that if the action ratio on the non-competitive venue is fixed, the whole action profile is fixed. Hence, for any binary-type \model{}, an action profile is in equilibrium if and only if $f(x) = 0$.

    \paragraph{Proof for properties 2 and 3} Now we consider researchers best responding to venue impacts. Suppose the current action profile has $f(\frac{\act_{H, 1}}{\act_{L, 1}}) < 0$. Fixing $\vvenue$ as a function of the current action profile, After best responding, the new action profile has $x' = \frac{\act'_{H, 1}}{\act'_{L, 1}}$ satisfying best-responding condition \Cref{eq:eqm-action-conf0}. Thus, it follows that $x' > x$. The impact on all venues will increase. The case for the current action profile $f(\frac{\act_{H, 1}}{\act_{L, 1}}) > 0$ can be proved in the same way.

    \paragraph{Proof for property 4 (convexity of the characteristic function)}
    Notice that $f(x) = \cost_{H, 1}^{\frac{1}{\alpha-1}}\sum_l h_l(x)$, where 
\begin{equation}\label{eq:uniqueness-proof-h}
    h_l(x)=\left[x\cdot\cost_{H, 1}\left(\frac{\cost_{H,l}}{\cost_{H, 1}}\right)^{\frac{\alpha}{\alpha-1}}-\cost_{L, 1}\cdot\left(\frac{\cost_{L,l}}{\cost_{L,1}}\right)^{\frac{\alpha}{\alpha-1}}\right]\venue_l^{\frac{\beta}{1-\alpha}}
\end{equation}
We calculate $h''_l(x)$:
\begin{align*}
        &h''_l(x)\\
=&\Tilde{\den}\cdot\beta(\frac{\cost_{L,l}}{\cost_{L,1}}/\frac{\cost_{H,l}}{\cost_{H,1}})^{\frac{1}{1-\alpha}}\cdot(\type-1) \\
&\cdot\left(\frac{1+\Tilde{\den}\cdot(\frac{\cost_{L,l}}{\cost_{L,1}}/\frac{\cost_{H,l}}{\cost_{H,1}})^{\frac{1}{1-\alpha}}\type x}{1+\Tilde{\den}\cdot(\frac{\cost_{L,l}}{\cost_{L,1}}/\frac{\cost_{H,l}}{\cost_{H,1}})^{\frac{1}{1-\alpha}}x}\right)^{\frac{\beta}{1-\alpha}}\\
       & \cdot\bigg[\cost_{L, 1}(\frac{\cost_{L,l}}{\cost_{L,1}})^{\frac{\alpha}{\alpha-1}}(\frac{\cost_{L,l}}{\cost_{L,1}}/\frac{\cost_{H,l}}{\cost_{H,1}})^{\frac{1}{1-\alpha}} \\
       & \cdot \Tilde{\den}((1+\type)(1-\alpha) +2\tilde{\den}(\frac{\cost_{L,l}}{\cost_{L,1}}/\frac{\cost_{H,l}}{\cost_{H,1}})^{\frac{1}{1-\alpha}}\type(1-\alpha)x)\\
   &+ \cost_{H, 1}(\frac{\cost_{H,l}}{\cost_{H,1}})^{\frac{\alpha}{\alpha-1}}\\
   & \cdot ((2-2\alpha) +(\frac{\cost_{L,l}}{\cost_{L,1}}/\frac{\cost_{H,l}}{\cost_{H,1}})^{\frac{1}{1-\alpha}}\tilde{\den}(1-\alpha)(1+\type)x)\\
   &+(-\cost_{L, 1}(\frac{\cost_{L,l}}{\cost_{L,1}})^{\frac{\alpha}{\alpha-1}}  +\cost_{H, 1}(\frac{\cost_{H,l}}{\cost_{H,1}})^{\frac{\alpha}{\alpha-1}}) \\
   & \cdot \tilde{\den}(\frac{\cost_{L,l}}{\cost_{L,1}}/\frac{\cost_{H,l}}{\cost_{H,1}})^{\frac{1}{1-\alpha}}\beta(\type-1)\bigg]\\
        &\bigg/\bigg[(-1+\alpha)^2\left(1+\Tilde{\den}\cdot(\frac{\cost_{L,l}}{\cost_{L,1}}/\frac{\cost_{H,l}}{\cost_{H,1}})^{\frac{1}{1-\alpha}}x\right)^2\\
        &\cdot\left(1+\Tilde{\den}\cdot(\frac{\cost_{L,l}}{\cost_{L,1}}/\frac{\cost_{H,l}}{\cost_{H,1}})^{\frac{1}{1-\alpha}}\type x\right)^2\bigg]
\end{align*}
Recall that $0 < \alpha < 1$ and $\theta > 1$. It's easy to check that all terms in $h''$ except for $(-\cost_{L, 1}(\frac{\cost_{L,l}}{\cost_{L,1}})^{\frac{\alpha}{\alpha-1}}+\cost_{H, 1}(\frac{\cost_{H,l}}{\cost_{H,1}})^{\frac{\alpha}{\alpha-1}})$ are positive. Hence we only need to show this remaining one is positive as well. By \Cref{assumption:increasing_cost_ratio}, we have $\frac{\cost_{L,l}}{\cost_{L,1}} > \frac{\cost_{H,l}}{\cost_{H,1}}$, so $(\frac{\cost_{L,l}}{\cost_{L,1}})^{\frac{\alpha}{\alpha-1}} < (\frac{\cost_{H,l}}{\cost_{H,1}})^{\frac{\alpha}{\alpha-1}}$. When $\cost_{H, 1}\geq \cost_{L, 1}$, $(-\cost_{L, 1}(\frac{\cost_{L,l}}{\cost_{L,1}})^{\frac{\alpha}{\alpha-1}}+\cost_{H, 1}(\frac{\cost_{H,l}}{\cost_{H,1}})^{\frac{\alpha}{\alpha-1}})$ is also always positive. Hence, $h''_l$ is the summation and multiplication of positive terms, so $h''_l > 0$, implying $f$ is convex.
\end{proof}

\subsection{The proof of \Cref{thm:scale-eq}}\label{appdx:scale-eq}
We show the effect of scaling the density of high-type researchers and low-type researchers non-uniformly. 
\begin{restatable}{theorem}{scaleeq}\label{thm:scale-eq}
    Under a binary-type  \model{}, with a non-competitive venue as in \Cref{assumption:non competitive venue}. 
    \begin{itemize}
   \item If the density of high-type researchers is scaled up, then the equilibrium impact of all venues will increase after scaling.  
        \item If the density of low-type researchers is scaled up, then the equilibrium impact of all venues will decrease after scaling.
    \end{itemize}
\end{restatable}
\begin{proof}
For the proof of \Cref{thm:scale-eq}, we write characteristic function as a function of two variables $x=\frac{\act_{H, 1}}{\act_{L, 1}}$ and $\tilde{\den} = \frac{\den_H}{\den_L}$, with the same formula as \Cref{def:stability function}. Suppose $x_0, \tilde{\den}_0$ and venue impacts $\vvenue(x_0, \tilde{\den}_0)$ satisfy $f(x_0, \tilde{\den}_0) = 0$, which are equilibrium outcomes. We begin by considering the case when the density of high-type researchers is scaled up. For any $\tilde{\den} > \tilde{\den}_0$,  $x = x_0\cdot\frac{\tilde{\den}_0 }{\tilde{\den}}$ keeps the same equilibrium impacts: $\vvenue(x, \tilde{\den}) = \vvenue(x_0, \tilde{\den}_0)$. Note here $x < x_0$. It suffices to show
\begin{equation*}
    f(x, \tilde{\den}) <0, 
\end{equation*}
i.e.\ $x$ is lower than the equilibrium action ratio. 

We notice that 
\begin{equation*}
    f(x, \tilde{\den})=\sum_l (\cost_{H, l})^{\frac{\alpha}{\alpha-1}} \cdot \confavg_{l}^{\frac{\beta}{1-\alpha}}\left(x-\ratiocost_l^{-\alpha}\right),
\end{equation*}
where $\ratiocost_j = (\frac{\cost_{L, j}}{\cost_{H, j}})^{\frac{1}{1-\alpha}}$. We want to show that 
\begin{equation*}
    f(x, \tilde{\den}) < f(x_0, \tilde{\den}_0) = 0
\end{equation*}

Fix any $l$, we have that 
\begin{align*}
    v_l(x_0, \den_0) &= \frac{1+b_j\cdot x_0 \cdot \tilde{\den}_0 \cdot \theta}{1+b_j\cdot x_0 \cdot \tilde{\den}_0 } \\
    &= \frac{1+b_j\cdot x \cdot \tilde{\den} \cdot \theta}{1+b_j\cdot x \cdot \tilde{\den}}\\
    &= v_l(x, \tilde{\den})
\end{align*}

Hence, under $x, \tilde{\den}$, $\venue_l$ holds fixed for all $l$. Therefore, to compare $f(x_0, \tilde{\den}_0) = 0$ and $f(x, \tilde{\den})$, we only need to inspect the scaling changes to $x-\ratiocost_l^{-\alpha}$. First, we observe that
\begin{equation*}
    \frac{x - \ratiocost_l^{-\alpha}}{x_0 - \ratiocost_l^{-\alpha}} = \frac{x}{x_0} + \frac{\ratiocost_l^{-\alpha}(\frac{x}{x_0} - 1)}{x_0 - \ratiocost_l^{-\alpha}}.
\end{equation*}
Whenever $x_0\geq - \ratiocost_l^{-\alpha}$, and recall that $x < x_0$, we have
\begin{equation*}
     \frac{x - \ratiocost_l^{-\alpha}}{x_0 - \ratiocost_l^{-\alpha}} = \frac{x}{x_0} + \frac{\ratiocost_l^{-\alpha}(\frac{x}{x_0} - 1)}{x_0 - \ratiocost_l^{-\alpha}} < \frac{x}{x_0}
\end{equation*}
Hence, the positive terms are scaled by a factor strictly smaller than $\frac{x}{x_0}$. Meanwhile, whenever $x_0\leq - \ratiocost_l^{-\alpha}$, 
\begin{equation*}
     \frac{x - \ratiocost_l^{-\alpha}}{x_0 - \ratiocost_l^{-\alpha}} = \frac{x}{x_0} + \frac{\ratiocost_l^{-\alpha}(\frac{x}{x_0} - 1)}{x_0 - \ratiocost_l^{-\alpha}} > \frac{x}{x_0}
\end{equation*}
Hence, the negative terms are scaled by a factor strictly larger than $\frac{x}{x_0}$. Thus, the positive terms in $f$ are scaled by a factor strictly lower than the negative parts. We know $f(x, \tilde{\den})<0$. The case for $\tilde{\den} < \tilde{\den}_0$ can be derived in the same way.
\end{proof}

\subsection{The proof of \Cref{observation:low-more-total-pub}}
\label{appdx:low-more-total-pub}

\lowmoretotalpub*
\begin{proof}
    We verify this observation by constructing such an example with binary researcher types, one has a higher impact than the other, and two venues. When $\type_H \gg \type_L$, and the cost is very high at the top conference, the high-type researcher may invest time in publishing on the top venue while the low-type researcher will focus on the easier one. Because of the high cost of publishing on the better venue, the high-type researcher may publish less in the end. 

    To construct a more natural example, we re-scale the time budget for all researchers to $40$, so we can interpret the total budget as the weekly time that a lab Principal Investigator can devote to research. Therefore, each cell in the cost matrix can be interpreted as the mean weekly hours required for publication efforts. Following the calculation which derives \Cref{prop:best-response}, the best-response action for each researcher will also be scaled by $40$. Let $\type_H = 20$, $\type_L = 1$, $\den_H = \frac{1}{3}$, $\den_L = \frac{2}{3}$, $\venue_1 = 1$, and $\venue_2 = 20$. Consider the following cost matrix:
    
\begin{table}[htbp]
\centering
\caption{Cost Matrix}
\label{tab:eq-low-more-pub}
\begin{tabular}{|c|c|c|}
\hline
 & Venue 1 & Venue 2 \\
\hline
$\type_H$ & $1$  & $15$  \\
\hline
$\type_L$ & $1$ & $40$  \\
\hline
\end{tabular}
\end{table}

By \Cref{thm:2p-uniqueness}, this \model{} will converge to the following unique equilibrium:

\begin{table}[htbp]
\centering
\caption{Actions in the equilibrium}
\label{tab:sampletable}
\begin{tabular}{|c|c|c|}
\hline
 & Venue 1 & Venue 2 \\
\hline
$\type_H$ & $12.03$ & $1.87$ \\
\hline
$\type_L$ & $14.19$ & $0.65$ \\
\hline
\end{tabular}
\end{table}
Calculation shows that $\vact_{L} \cdot \uvec > \vact_{H} \cdot \uvec$.
\end{proof}

\subsection{The proof of \Cref{{thm:top-venue-monotone}}}\label{appdx:top-venue-monotone}
\topvenuemonotone*
\begin{proof}
The proof directly compares different researcher types' actions on the top venue.

By \ref{prop:Quality is Monotone in difficulty}, we know that the venue with the highest impact is the venue with the greatest venue index $k$ when all researchers best respond. High type researcher will publish $\act_{i, k} = \frac{c_{i,k}^{\frac{1}{\alpha-1}} \venue_k^{\frac{\beta}{1-\alpha}}}{\vcost_i^{\frac{1}{\alpha-1}} \cdot \vvenue^{\frac{\beta}{1-\alpha}}}$, and low type researcher will publish $\act_{i', k} = \frac{c_{i',k}^{\frac{1}{\alpha-1}} \venue_k^{\frac{\beta}{1-\alpha}}}{\vcost_{i'}^{\frac{1}{\alpha-1}} \cdot \vvenue^{\frac{\beta}{1-\alpha}}}$. Comparing $\act_{i,k}$ and $\act_{i',k}$ is equivalent as comparing the following two terms:
\begin{align} \label{H-num-L-denom}
    &c_{i,k}^{\frac{1}{\alpha-1}} \cdot \venue_k^{\frac{\beta}{1-\alpha}} \cdot \vcost_{i'}^{\frac{1}{\alpha-1}} \cdot \vvenue^{\frac{\beta}{1-\alpha}} 
\end{align}
and 
\begin{align} \label{L-num-H-denom}
    &c_{i',k}^{\frac{1}{\alpha-1}} \cdot \venue_k^{\frac{\beta}{1-\alpha}} \cdot \vcost_i^{\frac{1}{\alpha-1}} \cdot \vvenue^{\frac{\beta}{1-\alpha}} 
\end{align}
Rearrange \eqref{H-num-L-denom} we obtain:
\begin{align} 
 \venue_k^{\frac{\beta}{1-\alpha}}  \cdot \sum_l  (\cost_{i,k})^{\frac{1}{\alpha-1}} \cdot (\cost_{i', l})^{\frac{\alpha}{\alpha-1}} \cdot \venue_l^{\frac{\beta}{1-\alpha}} \nonumber\\
= \venue_k^{\frac{\beta}{1-\alpha}} \cdot \sum_l (\cost_{i.k} \cdot \cost_{i',l})^{\frac{1}{\alpha-1}} \cdot \cost_{i',l}  \cdot \venue_l^{\frac{\beta}{1-\alpha}}\label{r-H-num-L-denom}
\end{align}
Rearrange \eqref{L-num-H-denom} we obtain:
\begin{align} 
 \venue_k^{\frac{\beta}{1-\alpha}} \cdot \sum_l  (\cost_{i',k})^{\frac{1}{\alpha-1}} \cdot (\cost_{i, l})^{\frac{\alpha}{\alpha-1}} \cdot \venue_l^{\frac{\beta}{1-\alpha}} \nonumber\\
 = 
\venue_k^{\frac{\beta}{1-\alpha}} \cdot \sum_l (\cost_{i'.k} \cdot c_{i,l})^{\frac{1}{\alpha-1}} \cdot \cost_{i,l} \cdot \venue_l^{\frac{\beta}{1-\alpha}}\label{r-L-num-H-denom}
\end{align}
By \Cref{assumption:increasing_cost_ratio}, we have $\cost_{i,k} \cdot c_{i',l} < c_{i',k} \cdot c_{i,l}$, $0< \alpha < 1$, hence $(c_{i,k} \cdot c_{i',l})^{\frac{1}{\alpha-1}} > (c_{i',k} \cdot c_{i,l})^{\frac{1}{\alpha-1}}$. Meanwhile, $c_{i',l} > c_{i,l}$ because $\theta^j > \theta^i$. Hence, we have $\Cref{r-H-num-L-denom} > \Cref{r-L-num-H-denom}$. This implies that 
$\Cref{H-num-L-denom} > \Cref{L-num-H-denom}$.
Rearrange the terms and we get:
$$\act_{i,k} = c_{i,k}^{\frac{1}{\alpha-1}} \cdot \venue_k^{\frac{\beta}{1-\alpha}} \cdot \vcost_{i'}^{\frac{1}{\alpha-1}} \cdot \vvenue^{\frac{\beta}{1-\alpha}}  > c_{i',k}^{\frac{1}{\alpha-1}} \cdot \venue_k^{\frac{\beta}{1-\alpha}} \cdot \vcost_i^{\frac{1}{\alpha-1}} \cdot \vvenue^{\frac{\beta}{1-\alpha}}  = \act_{i',k}$$
\end{proof}

\subsection{Empirical justification of spotlight signaling effect based on CVPR data}\label{appdx:empirical_justification}

We use empirical citation data on CVPR to justify our assumption that $\spadv(\spfrac_j)>1$. We fit the  spotlight signaling effect $\spadv(\spfrac_j)$ using empirical data. We use the average citation number to approximately assess the average impact of regular publications and spotlight publications. 

\begin{align}\label{citation-ratio-index}
   \spadv(\spfrac_j)= \frac{\text{average citation on spotlight of venue j in year t}}{\text{average citation on regular of venue j in year t}}.
\end{align}

Figure \ref{fig:cr-and-adver} illustrates the citation metrics against $\spfrac_j$ for the venue on Computer Vision and Pattern Recognition (CVPR) spanning 2014 to 2019.  It is observed that larger, multidisciplinary venues like NeurIPS, which often give preference to methodological and theoretical contributions for spotlight sessions, tend to receive fewer citations compared to applied research, thereby affecting the citation ratio index negatively. Therefore, we posit that venues like CVPR, focused on specialized fields, offer a citation ratio index that more accurately reflects the signaling effect of spotlight sessions. CVPR seems to be actively experimenting with their spotlight session during the period of investigation. They add spotlight sessions along with oral sessions in 2016, but later delete the spotlight session in 2019. The spotlight ratio  also varies a lot throughout the period,  with the spotlight ratio $1/\spfrac_j$ takes value from $11.9\%$ (2015) to $32.0\%$ (2016), with the average spotlight ratio being $23.8\%$. We consider $\log(\spfrac_j)$ to reduce the large variance between different years.

\begin{figure}[htbp]
\centering
\includegraphics[width=\columnwidth]{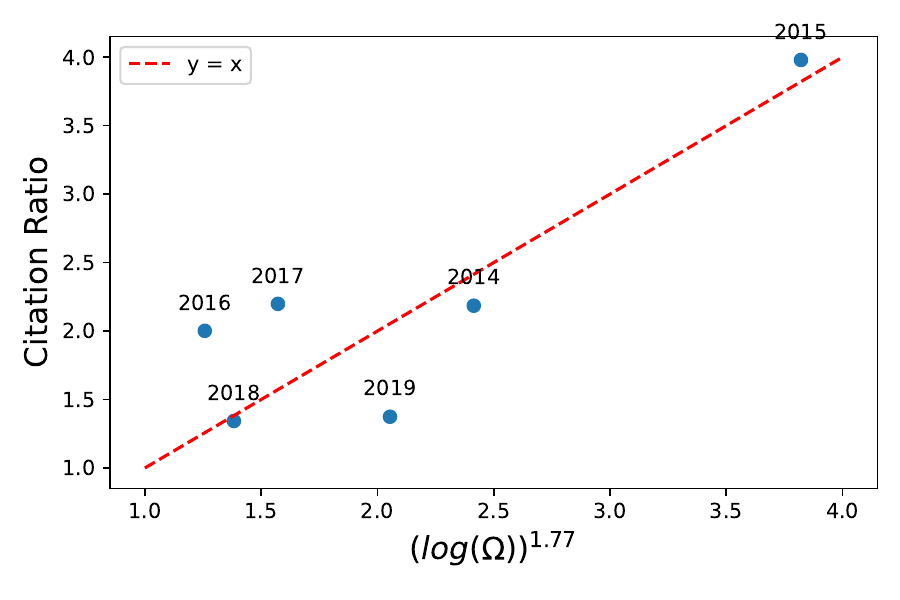}
\caption{CVPR Spotlight and Regular Session Citation Data Comparison}
\label{fig:cr-and-adver}
\end{figure}

Since we have limited data points, we do not draw conclusions on the form of $\spadv(\spfrac)$. We list this as a future work to do. One possible function form is $\spadv(\spfrac) = (\log(\spfrac))^p$, where $p > 0$. We find the best fit for the CVPR data is $p = 1.77$. 

\paragraph{Data Source} The list of CVPR accepted papers was sourced from the Digital Bibliography \& Library Project (DBLP), an extensive database in computer science. Citation data for each CVPR accepted paper was retrieved from Google Scholar using the service provided by SERP API, and the spotlight designation was confirmed through the CVPR program website. CVPR has both spotlight sessions and oral sessions, and we don't distinguish between them and mark both sessions as spotlight. We only include citation data whose citation number falls between $\mathsf{5th}$ and $\mathsf{95th}$ percentile each year. We then calculate the citation ratio index by \eqref{citation-ratio-index}. During the six years, the spotlight ratio $1/\spfrac_j$ takes value from $11.9\%$ (2015) to $32.0\%$ (2016), with the average spotlight ratio being $23.8\%$. For each of the six years studied, the citation ratio index consistently exceeds 1, supporting our model's assumption that $\spadv(\spfrac_j) > 1$.
Furthermore, our data suggest a general trend where $\spadv(\spfrac_j)$ increases as $\spfrac_j$ does. 

\subsection{The justification for \Cref{assumption:harder-spotlight}}\label{appdx:justification-harder-spotlight}

The venue organizer must choose the cost for spotlight session carefully. The following \Cref{thm:harder-spotlight} shows if publishing a spotlight paper is relatively the same hard as publishing a regular paper, the actual average impact of spotlight papers as a new venue will be the same as regular papers. This violates Constraint \ref{constraint:sp-better} that spotlight papers should gain more actual impact than regular papers on average.  \Cref{thm:harder-spotlight} motivates our \Cref{assumption:harder-spotlight} on feasible spotlight publication costs.

\begin{restatable}{proposition}{thmharderspotlight}
\label{thm:harder-spotlight}
    If for any two types $\type_i<\type_{i'}$, the cost ratio on spotlight paper remains the same as the cost ratio on regular session, i.e., $\frac{\cost_{i, j}^S}{\cost_{i', j}^S} = \frac{\cost_{i, j}}{\cost_{i', j}}$, then the spotlight papers and the regular papers have the same actual average impact.   
\end{restatable}
\begin{proof}
Define $\cost_{i, j}^S=\cost_{i, j}\cdot \spcostratio_{i, j}$, where $\spcostratio_{i, j}$ is the relative increase in cost for publishing a spotlight paper. If $\spcostratio_j$ are the same for all researchers, then we can factor out the cost ratio term and cancel them when calculating the average impact. Formally, $\forall j$ we have  
\begin{align*}
    \venue_{ j}^S
    &= \frac{(\vact_{:, j}^S \odot \vden)\cdot\vtype}{(\vact_{:, j}^S \odot \vden) \cdot\uvec} \\
    &=  \frac{(\vact_{:, j} \odot \vden) \cdot \spcostratio_{j}^{\frac{1}{\alpha-1}} \cdot \spadv(\spfrac_j)_j^{\frac{\beta}{1-\alpha}} \cdot \vtype}{(\vact_{:, j} \odot \vden) \cdot \spcostratio_{j}^{\frac{1}{\alpha-1}}\cdot \spadv(\spfrac_j)_j^{\frac{\beta}{1-\alpha}} \cdot \uvec} \\
    &= \frac{(\vact_{:, j} \odot \vden) \cdot  \vtype}{(\vact_{:, j} \odot \vden) \cdot\uvec} \\
    &= \venue_{j}
\end{align*}
Hence, the average impact on spotlight papers will be the same as regular papers without scaling. 
\end{proof}

Given \Cref{thm:harder-spotlight}, we introduce \Cref{assumption:harder-spotlight} as a natural design choice of venue organizers to ensure spotlight papers have a higher actual research impact than regular venue papers. \Cref{assumption:harder-spotlight} states that publishing a spotlight paper should be relatively harder for lower types than a regular paper. 

\subsection{Proof of \Cref{thm:2p-uniqueness}}
\label{apdx: proof 2p uniqueness}
\begin{proof}
Recall that we only need to prove the characteristic function has a unique zero point to show the uniqueness of equilibrium. We show $f(x)=0$ has a unique solution on $[0, \infty)$ by inspecting the geometry of the characteristic function $f(x)$. The following observations hold:
\begin{itemize}
    \item $f(0)<0, f(1)>0$, $\lim_{x\to \infty} f(x)=\infty$.
    \item $f(x)$ is convex on $x\geq 0$ by property 4 of  \Cref{lem:f-stability}.
\end{itemize}
Thus, $f(x) = 0$ admits a unique solution in $[0, \infty)$.
\end{proof}

\subsection{The proof of one-shot spotlight cost design}\label{appdx:one-shot-spotlight-cost}

\begin{restatable}
  {proposition}{leminvariantspcost}\label{lem:invariant-sp-cost}
Suppose for researcher of type $i$, the cost for publishing a regular paper on venue $j$ is fixed at $\cost_{i, j}$. Let the cost for publishing a spotlight paper be $\cost_{i, j}^S$. For venue $j$, fixing the cost $\cost_j$ of regular paper publications and the fraction $\frac{1}{\spfrac_j}$ of spotlight papers, the cost for publishing a spotlight paper can be set invariant of the venue average research impact and publication numbers. 
\end{restatable}
\begin{proof}
To prove the proposition, we first solve the researcher's utility maximization problem. Then we calculate the ratio of spotlight papers and regular papers and notice it's invariant of the dynamic venue average impact.

Define $\cost_{i, j}^S=\cost_{i, j}\cdot \spcostratio_{i, j}$, where $\spcostratio_{i, j}$ is the relative increase in cost for publishing a spotlight paper.

The researcher's utility maximization problem is as follows:
\begin{align*}
    \max_{\vact_i, \vact_{i}^S 
    } &\ \left( (\vact_{i})^{\alpha} \cdot \vvenue
    ^{\beta}+(\vact_{i}^S)^{\alpha} \cdot ( \spadv(\bm{\spfrac}) \odot \vvenue)^{\beta}\right)^{\frac{1}{\beta}} \\
\text{s.t.}\qquad & \vact_{i} \cdot \vcost_{i} + \vact_{i}^S \cdot \cost_{i}^S \leq 1  \nonumber
\end{align*}
where $\alpha \in (0,1)$ and $\beta \geq 1$.

Set up the Lagrangian and we get the following solutions:
\begin{align}\label{regular_optimal-action}
    \act_{i, j} &= \frac{(\cost_{i, j})^{\frac{1}{\alpha-1}} \cdot \venue_{j}^{\frac{\beta}{1-\alpha}}}
    {\sum_{l=1}^k \cost_{i,l}^{\frac{\alpha}{\alpha-1}} \cdot \venue_l^{\frac{\beta}{1-\alpha}} \cdot (1+r_{i,l}^{\frac{\alpha}{\alpha-1}} \cdot \spadv(\spfrac_l)^{\frac{\beta}{1-\alpha}})}  
\end{align}


\begin{align}\label{spotlight_optimal-action}
    \act_{i, j, S} &= 
    \frac{(\cost_{i, j, S})^{\frac{1}{\alpha-1}}  \cdot \spadv(\spfrac_j)^{\frac{\beta}{1-\alpha}}\cdot \venue_{j}^{\frac{\beta}{1-\alpha}}}   {\sum_{l=1}^k \cost_{i,l}^{\frac{\alpha}{\alpha-1}} \cdot \venue_l^{\frac{\beta}{1-\alpha}} \cdot (1+r_{i,l}^{\frac{\alpha}{\alpha-1}} \cdot \spadv(\spfrac_l)^{\frac{\beta}{1-\alpha}})}   \\
    &= \frac{(\cost_{i, j})^{\frac{1}{\alpha-1}} \cdot (\spcostratio_{i, j})^{\frac{1}{\alpha-1}} \cdot \spadv(\spfrac_j)^{\frac{\beta}{1-\alpha}}\cdot \venue_{j}^{\frac{\beta}{1-\alpha}}}   {\sum_{l=1}^k \cost_{i,l}^{\frac{\alpha}{\alpha-1}} \cdot \venue_l^{\frac{\beta}{1-\alpha}} \cdot (1+r_{i,l}^{\frac{\alpha}{\alpha-1}} \cdot \spadv(\spfrac_l)^{\frac{\beta}{1-\alpha}})}  \\
    &= a_{i, j} \cdot (\spcostratio_{i, j})^{\frac{1}{\alpha-1}} \cdot \spadv(\spfrac_j)^{\frac{\beta}{1-\alpha}}
\end{align}

Now, we can calculate the ratio of spotlight papers and regular papers and set it to $\frac{1}{\spfrac_j}$: 
\begin{align}\label{eq:fixed_ratio}
    \frac{1}{\spfrac_j} &= \frac{\vact_{:,j}^S \cdot\uvec}{\vact_{:, j}\cdot\uvec + \vact_{:,j}^S\cdot\uvec}\nonumber\\
    &= (\vcost_{:, j}^{\frac{1}{\alpha-1}} \odot \vden) \cdot (\vspcostratio_{: ,j })^{\frac{1}{\alpha-1}} \cdot \spadv(\spfrac_j)^{\frac{\beta}{1-\alpha}} \cdot (\venue_{j})^{\frac{\beta}{1-\alpha}}\nonumber\\
    &\bigg/\bigg[(\vcost_{:, j}^{\frac{1}{\alpha-1}} \odot \vden) \cdot\uvec\cdot (\venue_{j})^{\frac{\beta}{1-\alpha}}\\
    &+ (\vcost_{:, j}^{\frac{1}{\alpha-1}} \odot \vden) \cdot (\vspcostratio_{: ,j })^{\frac{1}{\alpha-1}} \cdot \spadv(\spfrac_j)^{\frac{\beta}{1-\alpha}} \cdot (\venue_{j})^{\frac{\beta}{1-\alpha}}\bigg]\nonumber\\
    &= \frac{(\vcost_{:, j}^{\frac{1}{\alpha-1}} \odot \vden) \cdot (\vspcostratio_{: ,j })^{\frac{1}{\alpha-1}} \cdot \spadv(\spfrac_j)^{\frac{\beta}{1-\alpha}}}{(\vcost_{:, j}^{\frac{1}{\alpha-1}} \odot \vden) \cdot\uvec + (\vcost_{:, j}^{\frac{1}{\alpha-1}} \odot \vden) \cdot (\vspcostratio_{: ,j })^{\frac{1}{\alpha-1}} \cdot \spadv(\spfrac_j)^{\frac{\beta}{1-\alpha}}} \label{eq:sp-frac-cost}
\end{align}
Note $\vspcostratio_{: ,j }$ is the fixed point of \Cref{eq:fixed_ratio}, which does not involve $\vvenue$;  hence $\spcostratio$ is invariant to the dynamic venue impact. 
\end{proof}

\subsection{The proof of \Cref{lem:eq-sp-unique}}\label{appdx:eq-sp-unique}
\lemequniquespotlight*
\begin{proof}
Define $\cost_{i, j}^S=\cost_{i, j}\cdot \spcostratio_{i, j}$, where $\spcostratio_{i, j}\geq 1$ is the relative increase in cost for publishing a spotlight paper.

The proof follows the same idea as \Cref{thm:2p-uniqueness}. Consider the same construction of characteristic function $f$ and decompose $f$ into linear combination of functions $h_l$. The venue $j$ with spotlight has 
\begin{align*}
        h_j(x)=\bigg[x\cdot\cost_{H, 1}\left(\frac{\cost_{H,j}}{\cost_{H, 1}}\right)^{\frac{\alpha}{\alpha-1}}(1+\spadv(\spfrac_j)^{\frac{\beta}{1-\alpha}}\spcostratio_{H}^{\frac{\alpha}{\alpha-1}})\nonumber\\
       \quad -\cost_{L, 1}\left(\frac{\cost_{L,j}}{\cost_{L,1}}\right)^{\frac{\alpha}{\alpha-1}}(1+\spadv(\spfrac_j)^{\frac{\beta}{1-\alpha}} \spcostratio_{L}^{\frac{\alpha}{\alpha-1}})\bigg]\venue_j^{\frac{\beta}{1-\alpha}}    
\end{align*}

By defining $\Tilde{\cost}_{i, 1}=\cost_{i, 1}((1+\spadv(\spfrac_j)^{\frac{\beta}{1-\alpha}}\spcostratio_{i}^{\frac{\alpha}{\alpha-1}}))$ for $i\in\{L, H\}$, we see that $\Tilde{\cost}_{H, 1}>\Tilde{\cost}_{L, 1}$ and the same form of $h_j$ in proof of \Cref{thm:2p-uniqueness} follows. Thus, $h_j$ is convex. 
\end{proof}

\subsection{The proof of \Cref{thm:sp-impact-compare}}\label{appdx:threshold_effect_proof}
\thmspimpactcompare*
\begin{proof}
    
We prove this by inspecting the extra cost each type invests on spotlight papers.  If a high type spends more budget on publishing spotlight papers than a low type, the equilibrium is the same as scaling down the high type's budget by a $\eta<1$ factor on regular venues, and also the same as scaling down the fraction of high types in the population. By \Cref{thm:scale-eq}, the research impact of all regular venues are reduced.

    Define $\cost_{i, j}^S =\cost_{i, j}\cdot \spcostratio_{i, j}$, where $\spcostratio_{i, j}$ is the relative increase in cost for publishing a spotlight paper.

First notice that the relative increase in cost on spotlight paper is fixed once $\spfrac_j$ is fixed by \Cref{eq:sp-frac-cost}.

    Now we analyze the equilibrium ratio of cost that each type invests on spotlight papers. Define $x=\frac{\act_{H, 1}}{\act_{L, 1}}$.
    \begin{align}\label{eq:cost-ratio-sp}
        \frac{\act_{H, j}^S\cdot\cost_{H, j}^S}{\act_{L, j}^S\cdot\cost_{L, j}^S}=\frac{\act_{H, j}\cost_{H, j}}{\cost_{L, j}\cost_{L, j}}(\frac{\spcostratio_{H, j}}{\spcostratio_{L, j}})^{\frac{\alpha}{\alpha-1}}=x(\frac{\cost_{H, j}}{\cost_{L, j}}\cdot \frac{\spcostratio_{H, j}}{\spcostratio_{L, j}})^{\frac{\alpha}{\alpha-1}}.
    \end{align}
    We would like to know if $\frac{\act_{H, j}^S\cdot\cost_{H, j}^S}{\act_{L, j}^S\cdot\cost_{L, j}^S}>1$ or not. Plugging $x=(\frac{\cost_{H, j}}{\cost_{L, j}}\cdot \frac{\spcostratio_{H, j}}{\spcostratio_{L, j}})^{\frac{\alpha}{1-\alpha}}$ into the function $f$ we defined in proof of \Cref{lem:eq-sp-unique}, if $f(x)>0$, then the equilibrium $x<(\frac{\cost_{H, j}}{\cost_{L, j}}\cdot \frac{\spcostratio_{H, j}}{\spcostratio_{L, j}})^{\frac{\alpha}{1-\alpha}}$ has $\frac{\act_{H, j}^S\cdot\cost_{H, j}^S}{\act_{L, j}^S\cdot\cost_{L, j}^S}<1$, implying all venues have increased average impact. Otherwise, $f(x)<0$ implies $\frac{\act_{H, j}^S\cdot\cost_{H, j}^S}{\act_{L, j}^S\cdot\cost_{L, j}^S}>1$, and all venues have decreased average impact.

     Consider the characteristic function $f(x)$. For venue $j$ switching to spotlight labeling, by definition,   
    \begin{align*}
        f(x) = \cost_{H, 1}^{\frac{1}{\alpha-1}}\bigg[\sum_l h_l(x)+[x\cdot\cost_{H, 1}\left(\frac{\cost_{H,l}}{\cost_{H, 1}}\right)^{\frac{\alpha}{\alpha-1}}\spcostratio_{H,l}^{\frac{\alpha}{\alpha-1}}\\
        -\cost_{L, 1}\left(\frac{\cost_{L,l}}{\cost_{L,1}}\right)^{\frac{\alpha}{\alpha-1}}\spcostratio_{L,l}^{\frac{\alpha}{\alpha-1}}]\spadv(\spfrac_l)^{\frac{\beta}{1-\alpha}}\venue_l^{\frac{\beta}{1-\alpha}}\bigg],
    \end{align*}
     where 
\begin{align*}
    h_l(x)=\left[x\cdot\cost_{H, 1}\left(\frac{\cost_{H,l}}{\cost_{H, 1}}\right)^{\frac{\alpha}{\alpha-1}}-\cost_{L, 1}\cdot\left(\frac{\cost_{L,l}}{\cost_{L,1}}\right)^{\frac{\alpha}{\alpha-1}}\right]\venue_l^{\frac{\beta}{1-\alpha}}
\end{align*}

Define the only differing term for different $j$ as
\begin{align*}
    &h_{j, S}(
x)\\
=&\bigg[x\cdot\cost_{H, 1}\left(\frac{\cost_{H,j}}{\cost_{H, 1}}\right)^{\frac{\alpha}{\alpha-1}}\spcostratio_{H,j}^{\frac{\alpha}{\alpha-1}}-\cost_{L, 1}\left(\frac{\cost_{L,j}}{\cost_{L,1}}\right)^{\frac{\alpha}{\alpha-1}}\spcostratio_{L,j}^{\frac{\alpha}{\alpha-1}}\bigg] \\
& \cdot \spadv(\spfrac_j)^{\frac{\beta}{1-\alpha}}\venue_j^{\frac{\beta}{1-\alpha}}.
\end{align*} 
When $x=(\frac{\cost_{H, j}}{\cost_{L, j}}\cdot \frac{\spcostratio_{H, j}}{\spcostratio_{L, j}})^{\frac{\alpha}{1-\alpha}}$, 
\begin{align*}
    h_{j, S}(x) =\cost_{H, 1}^{\frac{1}{1-\alpha}}[\cost_{L,j}^\frac{\alpha}{\alpha-1}\spcostratio_{L, j}^{\frac{\alpha}{\alpha-1}} -\cost_{L,j}^{\frac{\alpha}{\alpha-1}}\spcostratio_{L, j}^{\frac{\alpha}{\alpha-1}}]  \spadv(\spfrac_j)^{\frac{\beta}{1-\alpha}}\venue_j^{\frac{\beta}{1-\alpha}}=0.
\end{align*}
The value of $x$ is decreasing in venue index $j$, together with the fact that $h_l(x)$ and $h_{l,S}(x)$ are linear functions of $x$, implies that there exists a threshold index $j_0$, such that for $j\geq j_0$, $f(x)\leq  0$, while for  $j<j_0$, $f(x)>0$. 
\end{proof}